\documentclass[journal,twocolumn,twoside]{IEEEtran}
\usepackage{amsmath,amssymb, dsfont}
\usepackage{comment}

\usepackage{epsfig,latexsym,amssymb,amsmath,graphics,color}
\usepackage{graphicx,subfigure}
\usepackage[linesnumbered,ruled]{algorithm2e}
\usepackage{float}
\usepackage{verbatim}
\usepackage{multicol}
\usepackage{amstext}
\usepackage{mathrsfs}
\usepackage[all]{xy}
\usepackage{cite}
\usepackage{setspace}
\usepackage{amsthm}
\usepackage{multirow,multicol,array,booktabs}
\usepackage{slashbox}
\usepackage{array}
\usepackage{amsmath,amssymb, dsfont}
\usepackage{cases}
\usepackage{setspace}
\usepackage{subfig}

\usepackage{url}
\usepackage{amsmath,amssymb,amsfonts}
\usepackage{algorithmic}
\usepackage{textcomp}
\usepackage{xcolor}
\newtheorem{prop}{Proposition}

\newtheorem{cor}{Corollary}

\newtheorem{lm}{Lemma}

\newtheorem{thm}{Theorem}

\newcommand{\be}{\begin{eqnarray}}
\newcommand{\ee}{\end{eqnarray}}
\newcommand{\benn}{\begin{eqnarray*}}
\newcommand{\eenn}{\end{eqnarray*}}
\def\IR{\rm I \kern-0.20em R}

\newcommand{\bthm}{\begin{thm}}
\newcommand{\ethm}{\end{thm}}

\newcommand{\bcor}{\begin{cor}}
\newcommand{\ecor}{\end{cor}}
\newcommand{\bprop}{\begin{prop}}
\newcommand{\eprop}{\end{prop}}
\newcommand{\blm}{\begin{lm}}
\newcommand{\elm}{\end{lm}}
\newcommand{\beq}{\begin{equation}}
\newcommand{\eeq}{\end{equation}}
\newcommand{\ber}{\begin{eqnarray}}
\newcommand{\eer}{\end{eqnarray}}

\newcommand{\bproof}{\begin{proof}}
\newcommand{\eproof}{\end{proof}}

\newcommand{\bit}{\begin{itemize}}
\newcommand{\eit}{\end{itemize}}
\newcommand{\ben}{\begin{enumerate}}
\newcommand{\een}{\end{enumerate}}
\newcommand{\bdesc}{\begin{description}}
\newcommand{\edesc}{\end{description}}
\newcommand{\beqarrn}{\begin{eqnarray*}}
\newcommand{\eeqarrn}{\end{eqnarray*}}
\newcommand{\bproofof}{\begin{proofof}}
\newcommand{\eproofof}{\end{proofof}}
\newenvironment{rem}{\begin{trivlist}\item[]{\bf
Remark:}\hspace{4mm}}{\end{trivlist}}
\newcommand{\brem}{\begin{rem}}
\newcommand{\erem}{\end{rem}}
\newenvironment{rems}{\begin{trivlist}\item[]{\bf
Remarks}\begin{itemize}}{\end{itemize}\end{trivlist}}
\newcommand{\brems}{\begin{rems}}
\newcommand{\erems}{\end{rems}}
\newtheorem{fact}{Fact}
\newcommand{\bfact}{\begin{fact}}
\newcommand{\efact}{\end{fact}}
\newtheorem{examp}{Example}
\newcommand{\bexamp}{\begin{examp}\rm}
\newcommand{\eexamp}{\end{examp}}
\newtheorem{defn}{Definition}
\newcommand{\bdefn}{\begin{defn}\rm}
\newcommand{\edefn}{\end{defn}}

\newtheorem{alg}{Algorithm}
\newcommand{\balg}{\begin{alg}}
\newcommand{\ealg}{\end{alg}}

\newtheorem{prob}{Problem}
\newcommand{\bprob}{\begin{prob}}
\newcommand{\eprob}{\end{prob}}

\newcommand{\bvtm}{\begin{verbatim}}
\newcommand{\bfig}{\begin{figure}}
\newcommand{\efig}{\end{figure}}
\newcommand{\bcen}{\begin{center}}
\newcommand{\ecen}{\end{center}}

\long\def\comment#1{}

\def \n2{{N_0 \over 2}}

\def \h5{\hspace{0.5in}}

\renewcommand{\baselinestretch}{1}

\def\IR{\mathbb R}

\newtheorem{theorem}{Theorem}

\title{Indoor 3-Dimensional Visible Light Positioning: Error Metric and LED Layout Optimization}
\author{Jiaojiao Xu, Nuo Huang, and Chen Gong
	
	\thanks{This work was supported in part by the National Natural Science Foundation of China under Grant 62171428 and Grant 62101526, in part by Key Program of National Natural Science Foundation of China under Grant 61631018, in part by Key Research Program of Frontier Sciences of CAS under Grant QYZDY-SSW-JSC003, and in part by the Fundamental Research Funds for the Central Universities under Grant KY2100000118.}
	
	\thanks{Jiaojiao Xu, Nuo Huang, and Chen Gong are with Key Laboratory of Wireless-Optical Communications, Chinese Academy of Sciences, School of Information Science and Technology, University of Science and Technology of China, Hefei, China. Email: xjj1224@mail.ustc.edu.cn, \{huangnuo, cgong821\}@ustc.edu.cn.}}

\begin{document}
\maketitle{}

\begin{abstract}
	We consider $3$-dimensional (3D) visible light positioning (VLP) based on smartphone camera in an indoor scenario. Based on the positioning model in the quantized pixel-domain, we characterize the 3D normalized positioning error metric (NPEM) through the partial derivative of the positioning function, and evaluate the NPEM for horizontal and non-horizontal receiver camera positions. Moreover, under horizontal receiver terminal position, we explore the relationship between the NPEM and the light-emitting diode (LED) cell layout, approximate the relationship between the NPEM and the number of LEDs captured by the camera, and evaluate the approximation accuracy according to the simulated positioning error. Based on the approximation results, we optimize the LED transmitter cell layout to minimize NPEM assuming structured square cell layouts with certain distance parameters. 
\end{abstract}

\section{Introduction} \label{sec.Introduction}
Indoor positioning system (IPS) has attracted extensive attention due to its wide range of applications, for example, the positioning in museums and shopping malls. Up till now, plenty of IPSs have been proposed based on various technologies including Bluetooth~\cite{nieminen2014networking}, WiFi~\cite{yang2015wifi}, Radio Frequency Identification (RFID)~\cite{zou2014platform}, fingerprinting~\cite{kaemarungsi2005efficient}, Ultra-Wide Band (UWB)~\cite{zhang2010real}, and Infra-Red~\cite{hauschildt2010advances}. Bluetooth and WiFi-based IPSs generally suffer low positioning accuracy due to multi-path effects. RFID-based and fingerprinting-based IPSs perform positioning by matching the received signal with the information stored in the database, which requires a large data base and may lead to low position accuracy under small change of electromagnetic propagation environment. The above issue can be avoided by visible light positioning (VLP). In existing VLP system, one solution is to adopt a photodetector (PD) at the receiver~\cite{9011751,9080585,8486755,9008493,7339418}, which can estimate the transmission distance by detecting and analyzing the light properties such as received signal strength (RSS)~\cite{yasir2014indoor,9115244,7801028}, time of arrival (TOA) or time difference of arrival (TDOA)~\cite{do2014tdoa}, and phase of arrival (POA) or phase difference of arrival (PDOA)~\cite{panta2012indoor}. However, the RSS-based technique relies on the condition that the LED's transmitted power is accurately known and does not change over time, while both TOA/TDOA and POA/PDOA-based techniques require extremely accurate time/phase measurements~\cite{panta2012indoor}. In addition, it has been found that PDs are susceptible to the light beam direction, which may significantly limit the user mobility~\cite{zhao2016theoretical}. Thus, a promising alternative solution is to utilize an image sensor (IS)~\cite{wei2015high,liu2014towards,kuo2014luxapose,nakazawa2013indoor,nakazawa2014led}. Work~\cite{wei2015high} demonstrated a centimeter-level VLP system via simulations, where both PD and camera are employed. Work~\cite{nakazawa2013indoor} employed a camera equipped with a fish-eye lens to capture more LEDs and improve the positioning accuracy. 

In recent years, various high-precision methods and various positioning systems using smartphone camera as the receiver have been proposed~\cite{8125567,8417727,9397404,9322127,Zhu2019,9557845}. A detailed overview on the positioning approaches and the corresponding accuracy were presented in references~\cite{luo2017indoor} and~\cite{8292854}. In this work, we are no longer committed to a specific positioning approach, but pay more attention to the fundamental factors related to the positioning error. We consider a three-dimensional (3D) VLP system based on smartphone camera in an indoor scene, and utilize the transformations of different coordinate systems~\cite{9145234} to estimate the user position at any rotation angle. Then, we analyze the 3D normalized positioning error metric (NPEM), characterize the NPEM through the partial derivative of positioning function, and investigate the NPEM under horizontal and non-horizontal terminal positions. Moreover, we explore and approximate the relationship between the NPEM and the number of captured LEDs under parallel transmitter plane and receiver plane, and demonstrate the accuracy of the approximate relationship according to the simulated results. Finally, we propose a general form of LED cell layout optimization problem and optimize the LED cell layout parameters under square layout assumption to minimize the NPEM.

The remainder of this paper is organized as follows. Section~\ref{sec.2} introduces the translation and rotation model in camera imaging and 3D VLP. Section~\ref{sec.3} analyzes the NPEM and corresponding simulation results. Section~\ref{sec.4} explores the relationship between the NPEM and LED cell layout for parallel transmitter plane and receiver plane, evaluates the accuracy of approximated results according to the simulated positioning error, and explore the NPEM in infinite LED cell layout space under parallel and non-parallel transmitter plane and receiver plane. Section~\ref{sec.5} optimizes the LED transmitter layout to minimize the NPEM assuming structured square LED cell layout and provides the possible LED layout schemes. Finally, conclusion is made in Section~\ref{sec.Conclusions}. 

\section{Translation and rotation model For 3D VLP}\label{sec.2}
\subsection{VLP System with Four Coordinate Systems}\label{VLP System}
We consider a 3D VLP system, where multiple LEDs are adopted as anchor points with known positions, and a camera is adopted as the terminal to be positioned. The LEDs in the 3D space are displayed on the 2-dimensional (2D) pixel array of the terminal, which is utilized to estimate the camera center position. Such VLP system explores the geometric relations between the 3D LED positions in the real world and the 2D projection positions on the image plane. 

We consider $4$ coordinate systems for the camera-based VLP system, namely world coordinate system (WCS), camera coordinate system (CCS), image coordinate system (ICS), and pixel coordinate system (PCS)~\cite{9397404,8801944,9540364}. The dimensions of WCS and CCS are three, while the dimensions of ICS and PCS are two.

In the WCS, the camera center position to be estimated is denoted as $(U, V, W)$. In the CCS, such position is denoted as $(X, Y, Z)$, where the origin is located at the camera center; $x$ and $y$ axes are parallel to the phase plane; and $z$ axis is the lens optical axis. The imaging plane is located from the camera center with focal length $f$. In the ICS, the position in the CCS is projected onto the imaging plane, denoted as $(x^I, y^I)$. In the PCS, the image in the ICS is quantized into integer or half-integer pixels in the camera's CCD/CMOS chip, denoted as $(u^p, v^p)$. 

\subsubsection{The Coordinate Rotation} 
Assume that the rotation matrices are all calculated in the right-handed coordinate system, where counter-clockwise is the positive direction. Assume rotation angles $\theta_x, \theta_y$ and $\theta_z$ in the counter-clockwise manner around $x, y$ and $z$ axes from WCS to CCS, respectively. The rotation matrices around the three axes can be expressed as
\be\label{equ.Det001}
\begin{aligned} 
&{\bf R}_x(\theta_x)=
\begin{bmatrix} 
	1 & 0  & 0\\ 
	0 & \cos\theta_x  &-\sin\theta_x \\
	0 & \sin\theta_x  & \cos\theta_x 
\end{bmatrix},\ \\
&{\bf R}_{y}(\theta_y)=
\begin{bmatrix}
	\cos\theta_y  & 0  & \sin\theta_y\\
	0 & 1 & 0 \\ 
	-\sin\theta_y & 0  & \cos\theta_y 
\end{bmatrix}, \\
&{\bf R}_{z}(\theta_z)=
\begin{bmatrix} 
	\cos\theta_z &  -\sin\theta_z  & 0\\ 
	\sin\theta_z &  \cos\theta_z   & 0\\ 
	0 & 0 & 1 
\end{bmatrix}.
\end{aligned} 
\ee
The 3D rotation matrix ${\bf R}(\theta_x,\theta_y,\theta_z)$ is obtained by multiplying the above three matrices to the left according to the rotation order, given by
\be\label{equ.Det002}
\begin{aligned} 
{\bf R}\left( \theta _x,\theta _y,\theta _z \right) &={\bf R}_x\left( \theta _x \right) {\bf R}_{y}\left( \theta _y \right) {\bf R}_{z}\left( \theta _z \right)\\
   & \triangleq \begin{bmatrix} 
	R_{11} & R_{12} & R_{13}\\ 
	R_{21} & R_{22} & R_{23}\\ 
	R_{31} & R_{32} & R_{33} 
\end{bmatrix}. 
\end{aligned} 
\ee

\subsubsection{Translation and Rotation Model}
Consider the transformation from WCS to CCS, as shown in Figure~\ref{WCS2CCS}. Such transformation from WCS to CCS can be expressed as Equation~(\ref{equ.Det0000}),
\be\label{equ.Det0000}
\begin{aligned}
&\left[X\ Y\ Z\right]^T =\\
&\left[ \begin{matrix}
	R_{11}&	R_{12}&	R_{13}&		0\\
	R_{21}&	R_{22}&	R_{23}&		0\\
	R_{31}&	R_{32}&	R_{33}&		0\\
\end{matrix} \right] \left[ \begin{matrix}
	1&		0&		0&		-x^w_c  \\
	0&		1&		0&		-y^w_c\\
	0&		0&		1&		-z^w_c\\
	0&		0&		0&		1\\
\end{matrix} \right] \left[ \begin{array}{c}
	U\\
	V\\
	W\\
	1\\
\end{array} \right],
\end{aligned}
\ee
where $(x^w_c, y^w_c, z^w_c )$ is the origin coordinate of CCS in the WCS to be estimated. It is seen from Equation~(\ref{equ.Det0000}) that $[R_{11},R_{21},R_{31}]^T,\ [R_{12},R_{22},R_{32}]^T$ and $[R_{13},R_{23},R_{33}]^T$ are the coordinates of points $(1,0,0),\ (0,1,0)$ and $(0,0,1)$ of the WCS in the CCS, respectively.
\begin{figure}[htbp]
	\centering
	\includegraphics[width=8.5cm, height=4cm]{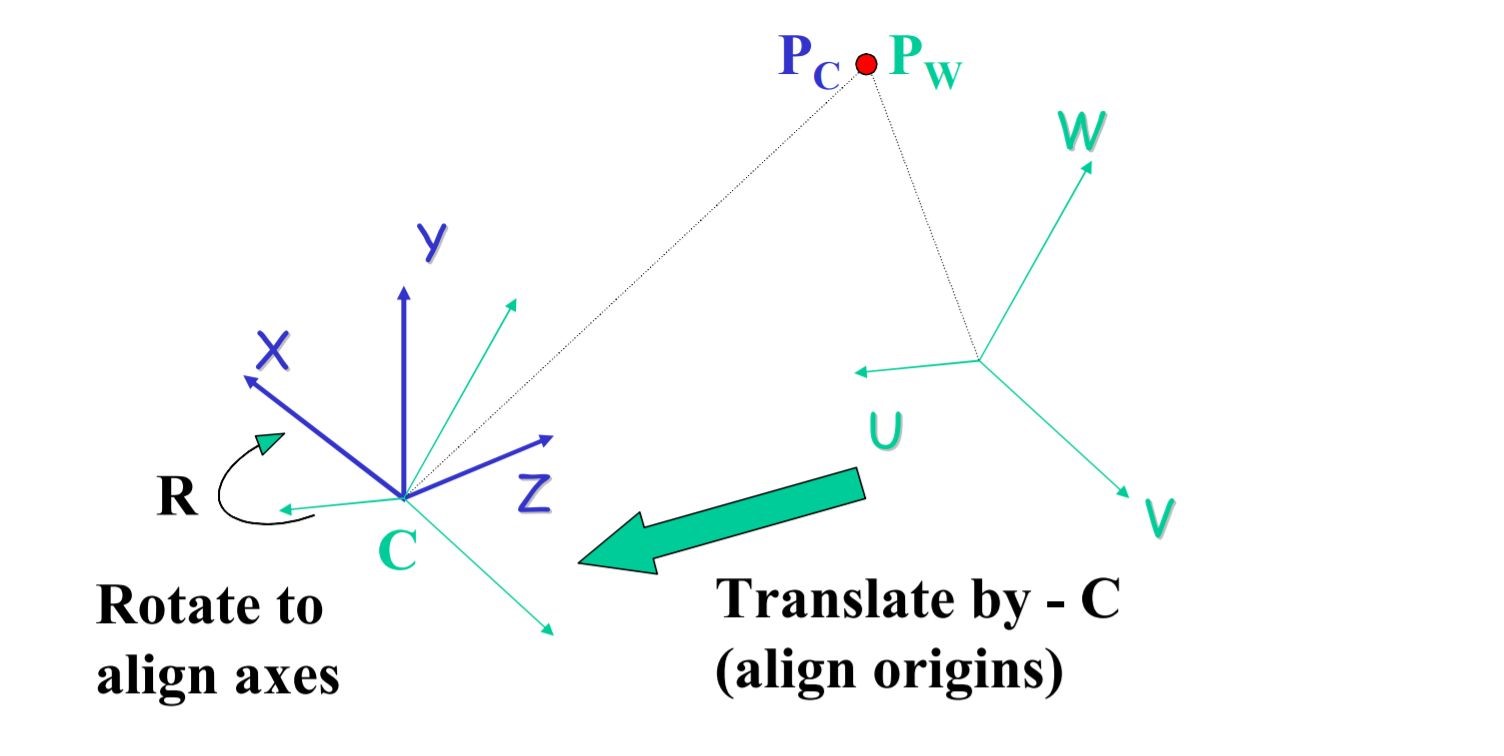}
	\caption{The transform from WCS to CCS.}\label{WCS2CCS}
\end{figure}

Figure~\ref{CCS2ICS} shows the transformation from position $(X, Y, Z)$ in the CCS to position $(x^I, y^I)$ in the ICS, derived as follows,
\begin{figure}[htbp]
	\centering
	\includegraphics[width=8.8cm, height=4cm]{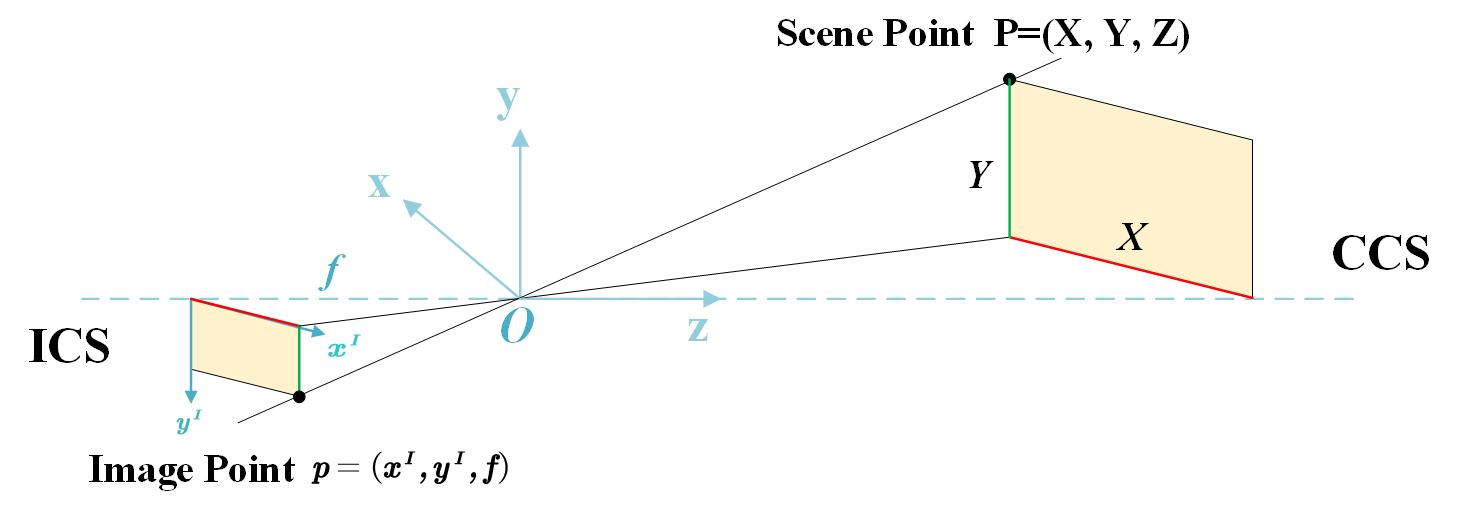}
	\caption{The projection from CCS to ICS.}\label{CCS2ICS}
\end{figure}
\be\label{equ.Det0001}
\begin{aligned}
	x^I=f\ \frac{X}{Z}, \ y^I=f\ \frac{Y}{Z}.
\end{aligned}
\ee

The transformation from ICS to PCS involves translation and scaling, as shown in Figure~\ref{ICS2PCS}. The ICS origin is the intersection of the camera's optical axis and the film plane, located at the image center point. The $x-$axis and $y-$axis are parallel to the $u-$axis and $v-$axis, respectively; and $(u_0,v_0)$ is the coordinate of the ICS origin in the PCS. Let $s_{x}$ and $s_{y}$ be the length per pixel in the $x-$axis and $y-$axis, respectively, with unit mm/pixel. Then, the quantization from ICS to PCS can be written as 
\begin{figure}[htbp]
	\centering
	\subfigure[Same coordinate dirction axis of xy-system and uv-system.]{
		\includegraphics[width=6cm, height=3.5cm]{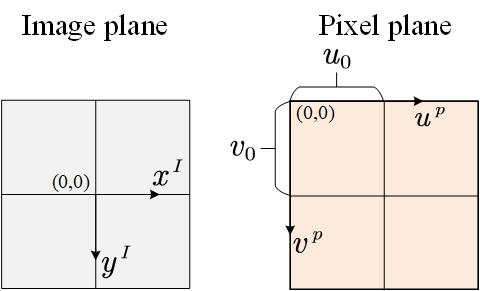}
	}\vspace{1ex}
	\subfigure[Opposite coordinate dirction axis of xy-system and uv-system.]{
		\includegraphics[width=6cm, height=3.5cm]{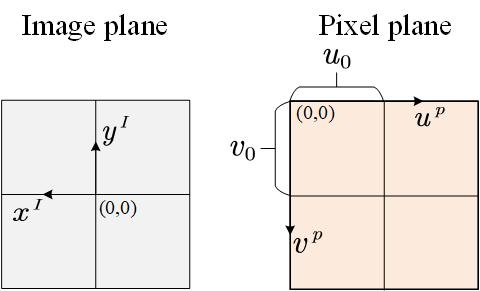}
	}
	\caption{The transformation from ICS to PCS.}\label{ICS2PCS}
\end{figure}
\be\label{equ.Det0002}
\begin{aligned}
	&\textbf{Q}[u^p]=\textbf{Q}[\frac{1}{s_x}x^I]+u_0=\textbf{Q}[\frac{1}{s_x}f\frac{X}{Z}]+u_0,  \\
	&\textbf{Q}[v^p]=\textbf{Q}[\frac{1}{s_y}y^I]+v_0=\textbf{Q}[\frac{1}{s_y}f\frac{Y}{Z}]+v_0,
\end{aligned}
\ee
where \textbf{Q}[$\cdot$] represents the pixel quantization operation to integer pixel values. Such operation introduces pixel quantization error and further leads to positioning error.

\subsection{3D VLP using Camera}
Let $L_{i}^{w}, L_{i}^{c}, L_{i}^{I}$ and $L_{i}^{p}$ denote the coordinates of LED $i$ in WCS, CCS, ICS and PCS, respectively ($i=1,2,3,\cdots$). Similarly, let $C^w, C^c, C^I$ and $C^p$ denote the corresponding coordinates of camera center in the four coordinate systems. The above coordinates and centers can be summarized as follows, 
\be\label{equ.Det003} 
\begin{aligned}       
	&\mathbf{WCS}: L_{i}^{w}=\left( x_{i}^{w},y_{i}^{w},z_{i}^{w}=h \right) ,\,C^w=\left( x_{\text{c}}^{w},y_{\text{c}}^{w},z_{\text{c}}^{w} \right);\\
	&\mathbf{CCS}: L_{i}^{c}=\left( x_{i}^{c},y_{i}^{c},z_{i}^{c} \right),\,C^c=\left( 0,0,0 \right);\\
	&\mathbf{ICS}: L_{i}^{I}=\left( x_i^I, y_i^I \right) ,\,C^I=\left( 0,0 \right);\\
	&\mathbf{PCS}: L_{i}^{p}=\left( u_{i}^{p},v_{i}^{p} \right) ,\,C^p=(u_0,v_0).   
\end{aligned}
\ee

The transformation from WCS to CCS is given by Equation~(\ref{equ.Det0000}); and the transformation from CCS to ICS is given by Equation~(\ref{equ.Det0001}). Based on the pixel coordinates in the PCS plane, the transformation from ICS to PCS is given by Equation~(\ref{equ.Det0002}).

By using Equation~(\ref{equ.Det002}) and substituting $x_{i}^{c}, y_{i}^{c}, z_{i}^{c}$ in the transformation from WCS to CCS, the above coordinate system transformations can be expressed as Equation~(\ref{equ.Det02}),
\begin{figure*}
\be\label{equ.Det02}
\begin{aligned}
	\left( u_{i}^{p}-u_0 \right)s_{x}+f \frac{R_{11}\left( x_{i}^{w}-x_{c}^{w} \right) +R_{12}\left( y_{i}^{w}-y_{c}^{w} \right) +R_{13}\left( z_{i}^{w}-z_{c}^{w} \right)}{R_{31}\left( x_{i}^{w}-x_{c}^{w} \right) +R_{32}\left( y_{i}^{w}-y_{c}^{w} \right) +R_{33}\left( z_{i}^{w}-z_{c}^{w} \right)}=0, \\
	\left( v_{i}^{p}-v_0 \right)s_{y}+f\frac{R_{21}\left( x_{i}^{w}-x_{c}^{w} \right) +R_{22}\left( y_{i}^{w}-y_{c}^{w} \right) +R_{23}\left( z_{i}^{w}-z_{c}^{w} \right)}{R_{31}\left( x_{i}^{w}-x_{c}^{w} \right) +R_{32}\left( y_{i}^{w}-y_{c}^{w} \right) +R_{33}\left( z_{i}^{w}-z_{c}^{w} \right)}=0,
\end{aligned}
\ee
\end{figure*}
where $(x_{c}^{w}, y_{c}^{w}, z_{c}^{w})$ is the camera center to be estimated. Assuming known LED coordinates $(x_{i}^{w}, y_{i}^{w}, z_{i}^{w})$ and pixel coordinates $(u_{i}^{p}, v_{i}^{p})$ for $i\ge 3$, the camera center position $(x^w_c, y^w_c, z^w_c)$ can be estimated along with rotation angles $\theta_x, \theta_y, \theta_z$. 

\section{Positioning error analysis and error metric simulation}\label{sec.3}
\subsection{Normalized Positioning Error Metric}
From Equation~(\ref{equ.Det02}), we have
\be\label{equ.Det901}
\begin{aligned}
&F_{x, i}(x_{\text{c}}^{w},y_{\text{c}}^{w},z_{\text{c}}^{w},\theta_x,\theta_y,\theta_z) \triangleq( u_{i}^{p}-u_0)s_x A_i+f B_i=0,\\  \vspace{-1ex}
&F_{y, i}(x_{\text{c}}^{w},y_{\text{c}}^{w},z_{\text{c}}^{w},\theta_x,\theta_y,\theta_z) \triangleq( v_{i}^{p}-v_0)s_y A_i+f C_i=0,
\end{aligned}
\ee
where 
$$A_i=R_{31}(x_{i}^{w}-x_{c}^{w}) +R_{32}(y_{i}^{w}-y_{c}^{w}) +R_{33}( z_{i}^{w}-z_{c}^{w}),$$ \vspace{-3ex} $$B_i=R_{11}(x_{i}^{w}-x_{c}^{w}) +R_{12}(y_{i}^{w}-y_{c}^{w}) +R_{13}(z_{i}^{w}-z_{c}^{w}),$$ \vspace{-3ex}
$$C_i=R_{21}(x_{i}^{w}-x_{c}^{w}) +R_{22}(y_{i}^{w}-y_{c}^{w}) +R_{23}(z_{i}^{w}-z_{c}^{w}).$$ 

To characterize the effect of pixel-domain quantization on the receiver positioning estimation, we first investigate its inverse, i.e., how the pixel-domain projection varies with the receiver position. Firstly, the partial derivative of $F_{x, i}$ with respect to $x_{c}^{w}$ is given by
\be\label{equ.Det902}
\frac{\textit{$\partial$}F_{x,i}}{\textit{$\partial$}x_{c}^{w}}=\frac{\textit{$\partial$}u_{i}^{p}}{\textit{$\partial$}x_{c}^{w}} s_x  A_i+( u_{i}^{p}-u_0) s_x (-R_{31})-f R_{11}=0.
\ee
Then, the pixel-domain projection varies with the receiver position, given by
\be\label{equ.Det903}
\frac{\textit{$\partial$}u_{i}^{p}}{\textit{$\partial$}x_{c}^{w}}=-\frac{1}{s_x A_i}\left[(u_{i}^{p}-u_0)s_x (-R_{31})-f R_{11}\right].
\ee
Denoting $M_i=B_i/A_i$ and $N_i=C_i/A_i$, based on Equations~(\ref{equ.Det901}) and (\ref{equ.Det903}), we have
\be\label{equ.Det1001}
\begin{aligned}
\frac{\textit{$\partial$}u_{i}^{p}}{\textit{$\partial$}x_{c}^{w}}=-\frac{f}{s_x A_i}[M_i  R_{31}-R_{11}],\\ \frac{\textit{$\partial$}v_{i}^{p}}{\textit{$\partial$}x_{c}^{w}}=-\frac{f}{s_y A_i}[N_i  R_{31}-R_{21}].
\end{aligned}
\ee
Assuming square pixel as that for common image sensor, we let $s_{x}=s_{y}$ in the following analysis. Without loss of generality, we normalize the pixel size to $1$ for numerical convenience. 

The partial derivative matrix in the pixel domain, denoted as $\bigtriangleup \textbf{C}^w$, is the partial derivative of $(u_{i}^{p},v_{i}^{p})$ for all $1 \leq i \leq n$ with respect to $x_{c}^{w}, y_{c}^{w}, z_{c}^{w}$. Assuming $n$ LEDs for positioning, the partial derivative matrix $\bigtriangleup \textbf{C}^w$ can be written as Equation~(\ref{equ.Det06}). Note that for horizontal receiver plane with $\theta_x=\theta_y=0$, we have ${\bf R}(\theta_x,\theta_y,\theta_z) ={\bf R}_{z}(\theta_z)$ and $A_i=z_{i}^{w}-z_{c}^{w}$, which is constant given receiver height $h$.
\begin{figure*}
\be\label{equ.Det06}
\begin{aligned}
	\begin{aligned}
		\bigtriangleup \textbf{C}^w&=\left[ 
		\begin{matrix}
			\frac{\textit{$\partial$}u_{1}^{p}}{\textit{$\partial$}x_{c}^{w}}&  \frac{\textit{$\partial$}u_{1}^{p}}{\textit{$\partial$}y_{c}^{w}}&  \frac{\textit{$\partial$}u_{1}^{p}}{\textit{$\partial$}z_{c}^{w}}\vspace{1ex} \\ 
			\frac{\textit{$\partial$}v_{1}^{p}}{\textit{$\partial$}x_{c}^{w}}&
			\frac{\textit{$\partial$}v_{1}^{p}}{\textit{$\partial$}y_{c}^{w}}& \frac{\textit{$\partial$}v_{1}^{p}}{\textit{$\partial$}z_{c}^{w}}\vspace{1ex}  \\
			&	\vdots &  \vspace{1ex}\\
			\frac{\textit{$\partial$}u_{i}^{p}}{\textit{$\partial$}x_{c}^{w}}&
			\frac{\textit{$\partial$}u_{i}^{p}}{\textit{$\partial$}y_{c}^{w}}&	
			\frac{\textit{$\partial$}u_{i}^{p}}{\textit{$\partial$}z_{c}^{w}}\vspace{1ex}  \\
			\frac{\textit{$\partial$}v_{i}^{p}}{\textit{$\partial$}x_{c}^{w}}&
			\frac{\textit{$\partial$}v_{i}^{p}}{\textit{$\partial$}y_{c}^{w}}&
			\frac{\textit{$\partial$}v_{i}^{p}}{\textit{$\partial$}z_{c}^{w}}\vspace{1ex} \\
			&	\vdots &  \vspace{1ex}\\
			\frac{\textit{$\partial$}u_{n}^{p}}{\textit{$\partial$}x_{c}^{w}}&
			\frac{\textit{$\partial$}u_{n}^{p}}{\textit{$\partial$}y_{c}^{w}}& 
			\frac{\textit{$\partial$}u_{n}^{p}}{\textit{$\partial$}z_{c}^{w}}\vspace{1ex}   \\
			\frac{\textit{$\partial$}v_{n}^{p}}{\textit{$\partial$}x_{c}^{w}}&
			\frac{\textit{$\partial$}v_{n}^{p}}{\textit{$\partial$}y_{c}^{w}}&	\frac{\textit{$\partial$}v_{n}^{p}}{\textit{$\partial$}z_{c}^{w}}  \\
		\end{matrix} \right]  \\ \vspace{2ex}
		&= -f \left[ \begin{matrix}
			\frac{1}{A_1}(M_1 R_{31}-R_{11})& \frac{1}{A_1}(M_1 R_{32}-R_{12})&	\frac{1}{A_1}(M_1 R_{33}-R_{13})\vspace{1ex}\\
			\frac{1}{A_1}(N_1 R_{31}-R_{21})& \frac{1}{A_1}(N_1 R_{32}-R_{22})&	\frac{1}{A_1}(N_1 R_{33}-R_{23})\vspace{1ex}\\
			&	\vdots &  \vspace{1ex}\\
			\frac{1}{A_i}(M_i R_{31}-R_{11})& \frac{1}{A_i}(M_i R_{32}-R_{12})&	\frac{1}{A_i}(M_i R_{33}-R_{13})\vspace{1ex}\\
			\frac{1}{A_i}(N_i R_{31}-R_{21})& \frac{1}{A_i}(N_i R_{32}-R_{22})&	\frac{1}{A_i}(N_i R_{33}-R_{23})\vspace{1ex}\\
			&	\vdots &  \vspace{1ex}\\
			\frac{1}{A_n}(M_n R_{31}-R_{11})& \frac{1}{A_n}(M_n R_{32}-R_{12})&	\frac{1}{A_n}(M_n R_{33}-R_{13})\vspace{1ex}\\
			\frac{1}{A_n}(N_n R_{31}-R_{21})& \frac{1}{A_n}(N_n R_{32}-R_{22})&	\frac{1}{A_n}(N_n R_{33}-R_{23})\\
		\end{matrix} \right].
	\end{aligned}
\end{aligned}
\ee
\end{figure*}

Performing singular value decomposition (SVD) on $\Delta \textbf{C}^w$, we have
\be\label{equ.Det08}
\begin{aligned}
	&\Delta \textbf{C}^w=\textbf{U} \boldsymbol{\Sigma} \textbf{V}^H,\ \text{rank}(\Delta \textbf{C}^w) =r, \\
	&\textbf{U}=[\textbf{u}_1,\textbf{u}_2,...,\textbf{u}_m],\ \textbf{V}=[\textbf{v}_1,\textbf{v}_2,...,\textbf{v}_n],\ \\
	&\boldsymbol{\Sigma}= \left[ \begin{matrix}
		\boldsymbol{\Sigma_1}&	\textbf{O}\\
		\textbf{O}&	    \textbf{O}\\	
	\end{matrix} \right],	
\end{aligned}
\ee
where $\boldsymbol{\Sigma_1} =\text{diag}(\text{$\sigma$}_1,\text{$\sigma$}_2,...,\text{$\sigma$}_{r})$ with all positive singular values and $r \leq 3$.

Assuming independent pixel-domain quantization error, we investigate the WCS positioning error due to the quantization error. Denote Moore-Penrose (MP) inverse $(\Delta \textbf{C}^w)^{+}=\textbf{V}\boldsymbol{{\Sigma}}^{+} \textbf{U}^H$, where $\boldsymbol{{\Sigma}}^{+}=[[\boldsymbol{{\Sigma_1}^{-1}}, \textbf{O}]^T,[\textbf{O},\textbf{O}]^T]$. Let $\textbf{e}\in {\mathbb {R}}^{m\times1}$ be the pixel-domain quantization error, satisfying $E[\textbf{ee}^H] \triangleq \beta \bf{I_m}$ for certain constant $\beta$ due to the independent quantization error assumption, where constant $\beta$ is related to the camera and its pixel plane characteristics.

Given quantization error $\textbf{e}$, we adopt least-square (LS) criterion and find the solution with the least norm, which can be obtained via multiplying MP inverse $(\Delta \textbf{C}^w)^{+}$, as the performance metric of positioning error. The expectation of its Frobenius norm is given by
\be\label{equ.Det00}
\begin{aligned}
	\textbf{E}\left[\parallel{(\Delta \textbf{C}^w)^{+} \textbf{e}}\parallel_F\right] &=\textbf{E}\left[\parallel\textbf{V}\boldsymbol{{\Sigma}^{+}} \textbf{U}^H\textbf{e}\parallel _{F}\right]\\ &=\textbf{E}\left[\parallel\boldsymbol{{\Sigma}^{+}} \textbf{U}^H\textbf{e}\parallel _{F}\right]\\
	&=\textbf{E}\left[\sqrt{\text{tr}\{\boldsymbol{\Sigma}^{+}\textbf{U}^H \textbf{e}(\boldsymbol{\Sigma}^{+} \textbf{U}^H \textbf{e})^H\}}\right]\\
	&=\sqrt{\text{tr}\{\textbf{E}\left[\boldsymbol{\Sigma}^{+} \beta \bf{I_m} (\boldsymbol{\Sigma}^{+})^H\right]\}}\\
	&=\sqrt{\beta \cdot \sum_{i=1}^r{\frac{1}{{\text{$\sigma$}_{\text{i}}^{2}}}}}.
\end{aligned}
\ee

Accordingly, we define the normalized positioning error metric as NPEM, given by Equation~(\ref{equ.Det09}) for camera-based positioning,
\be\label{equ.Det09}
\text{NPEM} \triangleq \sqrt{\beta \cdot \sum_{i=1}^r{\frac{1}{{\text{$\sigma$}_{\text{i}}^{2}}}}}.
\ee

\subsection{NPEM of $3D$ Rotation in $O$-$xyz$ Space}\label{3D error perf}
Consider the 3D rotation in the $O$-$xyz$ space, where rotation angles $(\theta_x, \theta_y, \theta_z)$ and position $(x_{c}^w,y_{c}^w,z_{c}^w)$ are both unknown with $\theta_x\in (-\frac{\pi}{2},\frac{\pi}{2}),\ \theta_y\in (-\frac{\pi}{2}, \frac{\pi}{2})$ and $\theta_z\in (0,2\pi )$.

Figure~\ref{lednumbercaptured} shows the positioning configuration under rotation matrix ${\bf R}(\theta_x,\theta_y,\theta_z)$. The coverage area of a camera is a cone with $(x_{c}^w,y_{c}^w,z_{c}^w)$ as the vertex and the receiver FOV as the radiation angle. Then, the LEDs captured by the smartphone camera are those contained within the intersection of the conical surface and the transmitter plane.
\begin{figure}[htbp]
	\centering
	\includegraphics[width=4cm, height=5cm]{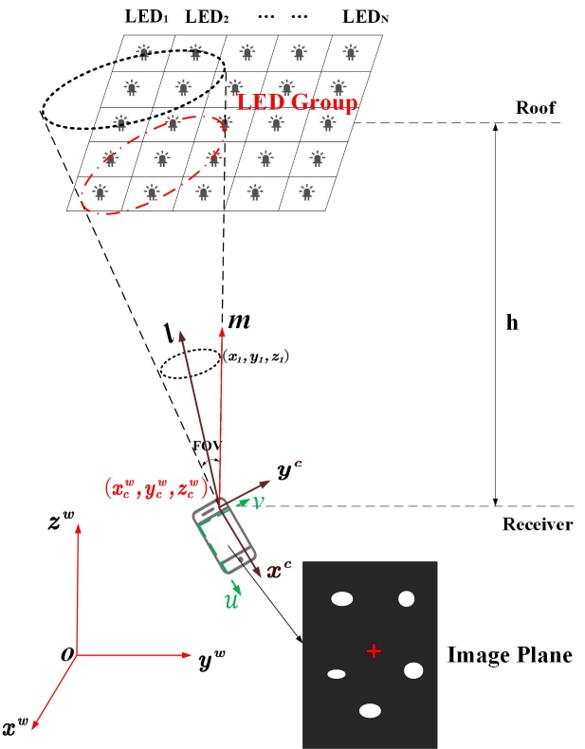}
	\caption{The LEDs captured by the smartphone camera (totally $5$ LEDs).}\label{lednumbercaptured}
\end{figure}

Consider the following parameter settings: $5$m$\times5$m$\times3$m room, LED height from the ground $z_{i}^{w}=2.75$\ m, $FOV=90$ deg, $\beta= 1.7857$~\cite{8644462,8519633}. Consider the density of LED cell layouts as $1$ center LED/m$^2$, with totally $25$ LED transmitters uniformly distributed at positons $(i - 0.5, j - 0.5)$ for $1 \leq i,j \leq 5$. Figure~\ref{(2.5,2.5,1)(10,10,50)} shows the number of captured LEDs and the NPEM with rotation angles $\theta_x\in (-\frac{\pi}{2},\frac{\pi}{2}),\ \theta_y\in (-\frac{\pi}{2}, \frac{\pi}{2})$ and $\theta_z\in (0,2\pi )$ under angle interval $(\Delta\theta_x, \Delta\theta_y, \Delta\theta_z)= (\frac{\pi}{18},\frac{\pi}{18},\frac{5\pi}{18})$ at camera center $(x_{c}^w,y_{c}^w,z_{c}^w)=(2.5,2.5,1)$. It can be seen that the number of captured LEDs is symmetric on the rotation angle due to symmetric layout, and the minimum NPEM direction changes at $\theta_z=90, 180, 270$, which shows that at a certain receiver point, the rotation angle with the minimum NPEM is not necessarily parallel to the horizontal plane.
\begin{figure}[htbp]
	\centering
	\subfigure[LED number.]{
		\includegraphics[width=2.7cm, height=2cm]{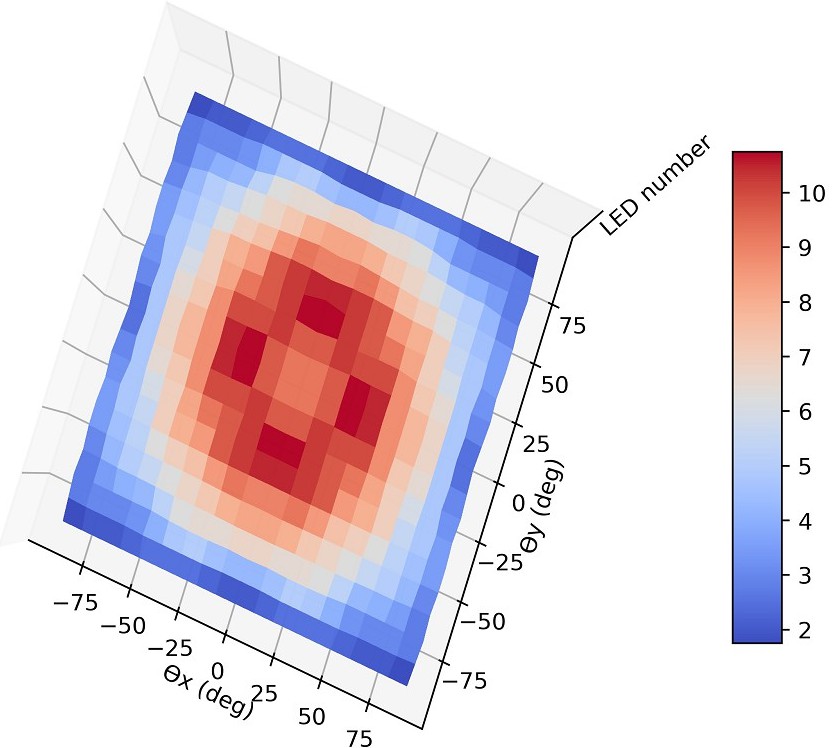}
	}\hspace{-1ex}
	\subfigure[$\theta_z$=0]{
		\includegraphics[width=2.7cm, height=2cm]{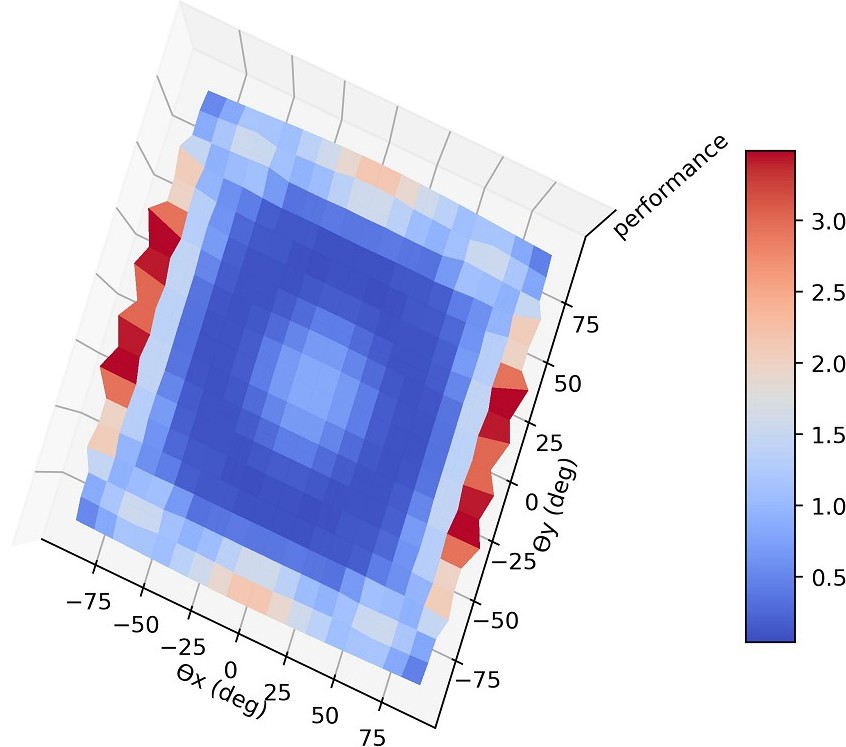}
	}
	\hspace{-1ex}
	\subfigure[$\theta_z$=50]{
		\includegraphics[width=2.7cm, height=2cm]{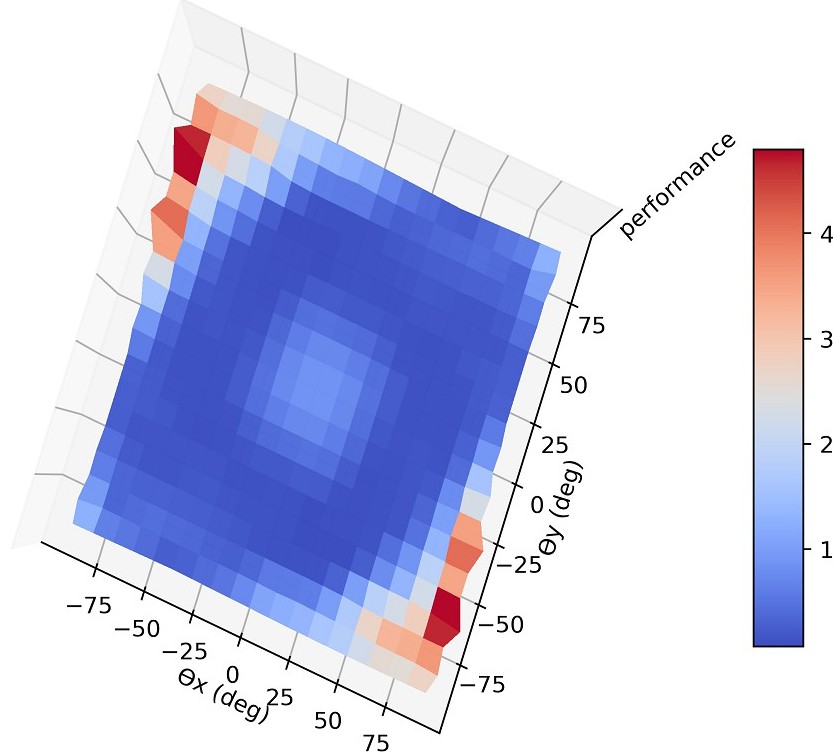}
	}\\ \vspace{1ex}
	\hspace{-1ex}
	\subfigure[$\theta_z$=100]{
		\includegraphics[width=2.7cm, height=2cm]{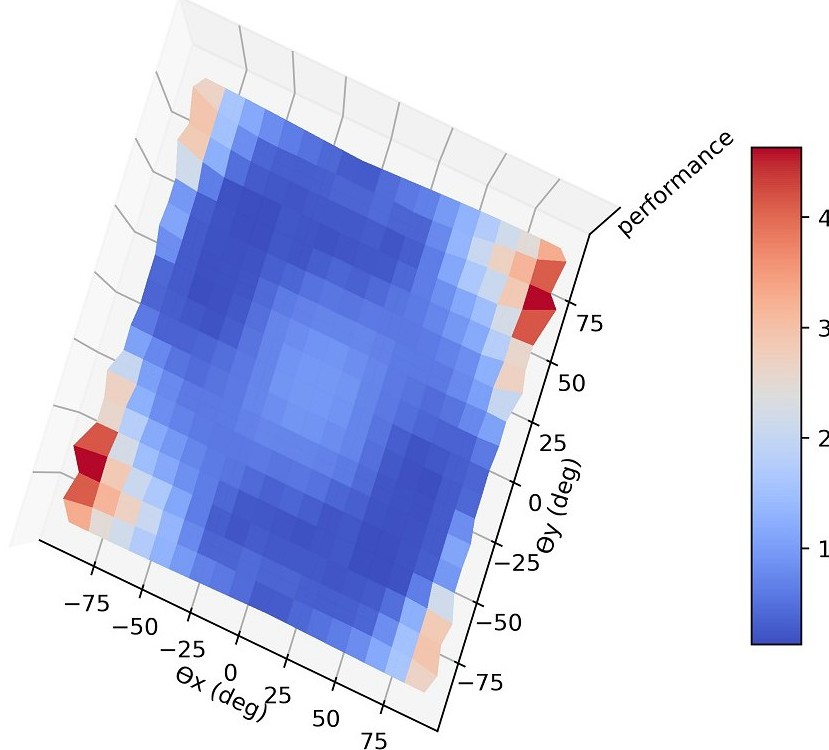}
	}
	\hspace{-1.2ex}
	\subfigure[$\theta_z$=150]{
		\includegraphics[width=2.7cm, height=2cm]{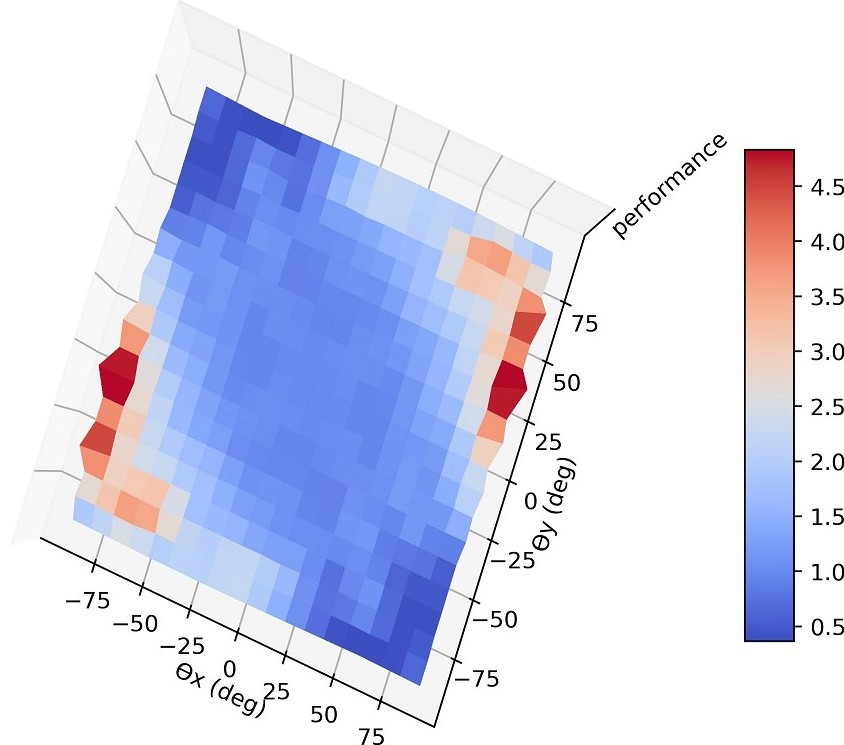}
	}
	\hspace{-1.2ex}
	\subfigure[$\theta_z$=200]{
		\includegraphics[width=2.7cm, height=2cm]{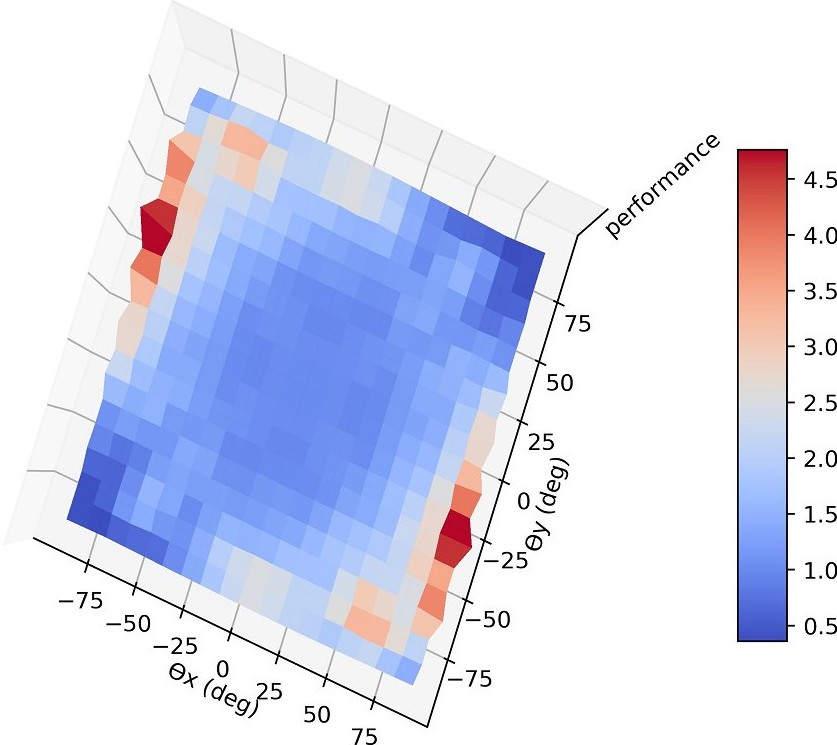}
	}\\ \vspace{1ex}
	\hspace{-1ex}
	\subfigure[$\theta_z$=250]{
		\includegraphics[width=2.7cm, height=2cm]{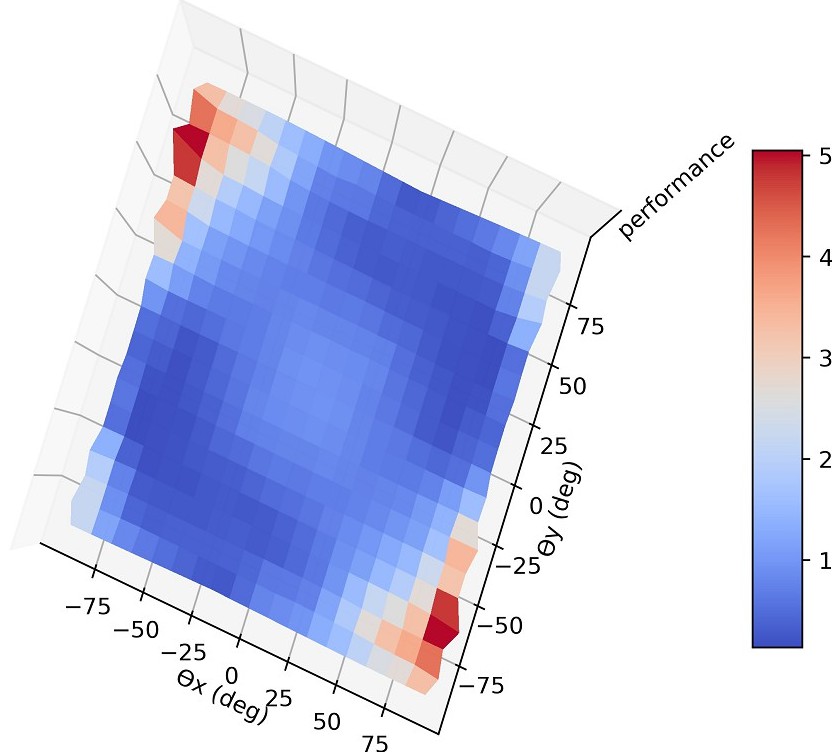}
	}
	\hspace{-1.2ex}
	\subfigure[$\theta_z$=300]{
		\includegraphics[width=2.7cm, height=2cm]{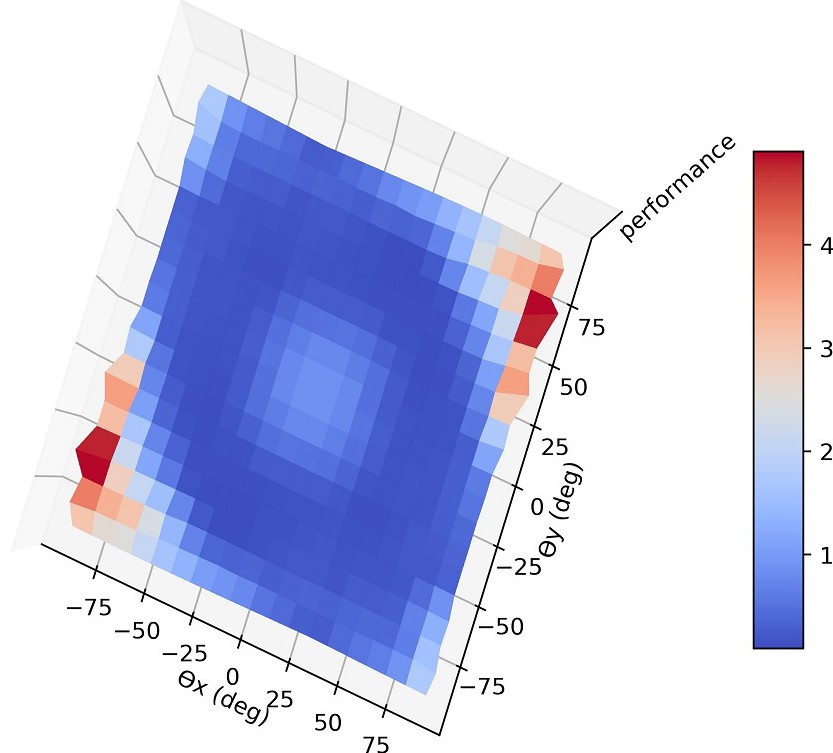}
	}
	\hspace{-1.2ex}
	\subfigure[$\theta_z$=350]{
		\includegraphics[width=2.7cm, height=2cm]{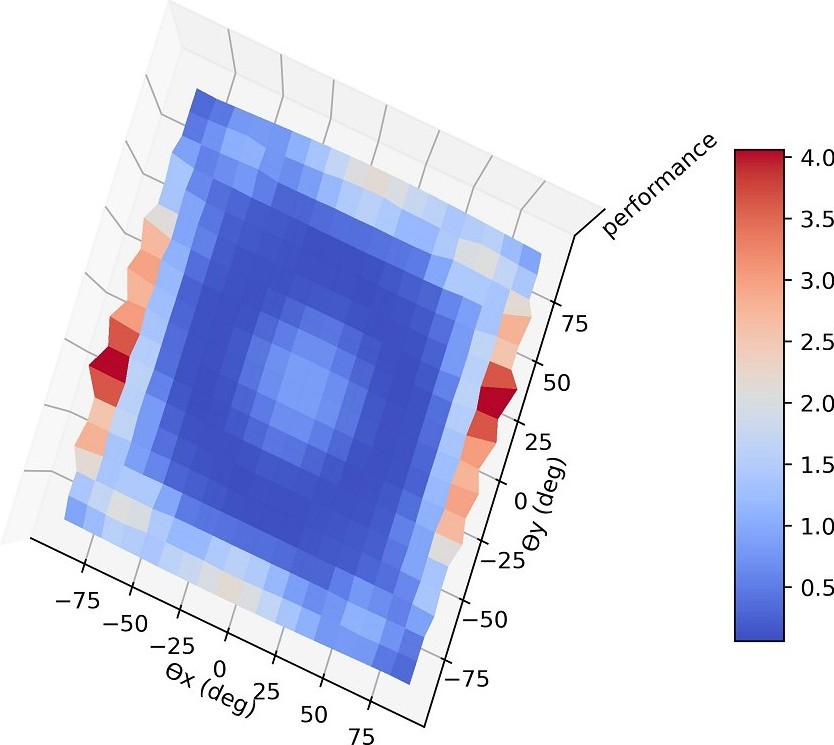}
	}
	\caption{The number of captured LEDs and the NPEM results in $3D$ O-xyz Rotation.}\label{(2.5,2.5,1)(10,10,50)}
\end{figure}
\subsection{NPEM of $2D$ Rotation in $xoy$-Plane}\label{xoy-Plane}
Considering the 2D rotation in the $xoy$-plane with multiple LEDs, we have the following theorem.
\begin{theorem}
	NPEM is constant with $\theta_z$ under the 2D $xoy$-plane rotation for parallel transmitter plane and receiver plane.
\end{theorem}
\begin{proof}
Since the transmitter plane is parallel to the horizontal receiver plane, we have $\theta_x=\theta_y=0$ and ${\bf R}_x(\theta _x) ={\bf R}_{y}(\theta _y) = \textbf{I}_3$. Then, the 3D rotation matrix is recast as
\be\label{equ.Det904}
{\bf R}(\theta_x,\theta_y,\theta_z) = {\bf R}_z(\theta_z)=\left[ \begin{matrix}
\cos\theta _z&		-\sin\theta _z&		0\\
\sin\theta _z&		\cos\theta _z&		0\\
0&		0&		1\\
\end{matrix} \right].
\ee
Then, $M_i$ and $N_i$ in Equation~(\ref{equ.Det06}) are given by
\be\label{equ.Det905}
\begin{aligned}	
 M_i&=\frac{R_{11}\left( x_{i}^{w}-x_{c}^{w} \right) +R_{12}\left( y_{i}^{w}-y_{c}^{w} \right)}{\left( z_{i}^{w}-z_{c}^{w} \right)}\\
   &=-(R_{11}m_i+R_{12}n_i), \\ 
 N_i&=\frac{R_{21}\left( x_{i}^{w}-x_{c}^{w} \right) +R_{22}\left( y_{i}^{w}-y_{c}^{w} \right)}{\left( z_{i}^{w}-z_{c}^{w} \right)}\\
   &=-(R_{21}m_i+R_{22}n_i),   
\end{aligned}
\ee
where $m_i=-\frac{\left( x_{i}^{w}-x_{c}^{w} \right)}{\left( z_{i}^{w}-z_{c}^{w} \right)}$ and $n_i=-\frac{\left( y_{i}^{w}-y_{c}^{w} \right)}{\left( z_{i}^{w}-z_{c}^{w} \right)}$. Thus, Equation~(\ref{equ.Det06}) can be rewritten as
\be\label{equ.Det15}
\bigtriangleup \textbf{C}^w=\frac{-f}{z_{i}^{w}-z_{c}^{w}} \left[ \begin{matrix}
	R_{11}&		\,\,R_{12}&		R_{11}m_1+R_{12}n_1\\
	R_{21}&		\,\,R_{22}&		R_{21}m_1+R_{22}n_1\\
	& \vdots &\\
	R_{11}&		\,\,R_{12}&		R_{11}m_i+R_{12}n_i\,\\
	R_{21}&		\,\,R_{22}&		R_{21}m_i+R_{22}n_i\\
	& \vdots &\\
	R_{11}&		\,\,R_{12}&		R_{11}m_n+R_{12}n_n\,\\
	R_{21}&		\,\,R_{22}&		R_{21}m_n+R_{22}n_n\\
\end{matrix} \right].
\ee

Using Equation~(\ref{equ.Det15}), we have Equation~(\ref{equ.Det17}).
\begin{figure*}
\be\label{equ.Det17}
\left( \Delta \textbf{C}^w \right) ^H \cdot \Delta \textbf{C}^w=\left[\frac{f}{z_{i}^{w}-z_{c}^{w}}\right]^2 \cdot \left[ \begin{matrix}
	n&		0&		\sum_{i=1}^n{m_i}\vspace{1ex}\\
	0&		n&		\sum_{i=1}^n{n_i}\vspace{1ex}\\
	\sum_{i=1}^n{m_i}&		\sum_{i=1}^n{n_i}&		\sum_{i=1}^n{\left( m_{i}^{2}+n_{i}^{2} \right)}\\
\end{matrix} \right]. 
\ee
\end{figure*}
Furthermore, based on the relationship of eigenvalue of matrix $(\Delta {\mathbf C}^w)^H\cdot \Delta{\mathbf C}^w$ and singular value $\sigma_i$ of matrix $\Delta{\mathbf C}^w$, we have Equation~(\ref{equ.Det18}),
\begin{figure*}
\be\label{equ.Det18}
\sum_{i=1}^3{\frac{1}{\sigma^2_i}=[\frac{z_{i}^{w}-z_{c}^{w}}{f}]^2 \cdot \left( \frac{2}{\sum_{i=1}^n{\left( m_{i}^{2}+n_{i}^{2} \right)}-2r_n+n}+\frac{2}{\sum_{i=1}^n{\left( m_{i}^{2}+n_{i}^{2} \right) +2r_n+n}}+\frac{1}{n} \right)},
\ee
\end{figure*}
where $$2r_n=\bigg{ \{ }\sum_{i=1}^n{\left( m_{i}^{4}+n_{i}^{4} \right)}+2\sum_{i\ne j}^n{\left( m_{i}^{2}m_{j}^{2}+n_{i}^{2}n_{j}^{2} \right)}$$ $$+2\sum_{i,j=1}^n{m_{i}^{2}n_{j}^{2}}-2\left( n-2 \right) \sum_{i=1}^n{\left( m_{i}^{2}+n_{i}^{2} \right)}$$ $$+8\sum_{i\ne j}^n{\left( m_im_j+n_in_j \right)}+n^2\bigg{ \} }^{1/2}.$$

According to Equations~(\ref{equ.Det09}) and (\ref{equ.Det18}), the NPEM does not depend on $\theta_z$ for parallel transmitter plane and receiver plane.
\end{proof}

For the rotation in 2D $xoy$-plane, we have ${\bf R}(\theta_x,\theta_y,\theta_z)={\bf R}_{z}(\theta_z)$. Then, Equation~(\ref{equ.Det06}) can also be simplified using pixel coordinates $(u_{i}^{p},v_{i}^{p})$, given by
\be\label{equ.Det10}
\begin{aligned}
	\begin{aligned}
		\bigtriangleup \textbf{C}^w = \frac{f}{z_{i}^{w}-z_{c}^{w}} \left[ 
		\begin{matrix}
			\cos \theta _z&	 -\sin \theta _z&	\left( u_{1}^{p}-u_0 \right) s_x /f  \\
			\sin \theta _z&	\,\,\cos \theta _z&	\left( v_{1}^{p}-v_0 \right) s_y /f  \\
			& \vdots & \\
			\cos \theta _z&	 -\sin \theta _z&	\left( u_{i}^{p}-u_0 \right) s_x /f  \\
			\sin \theta _z&	\,\,\cos \theta _z&	\left( v_{i}^{p}-v_0 \right) s_y /f  \\
			& \vdots & \\
			\cos \theta _z&	 -\sin \theta _z&	\left( u_{n}^{p}-u_0 \right) s_x /f  \\
			\sin \theta _z&	\,\,\cos \theta _z&	\left( v_{n}^{p}-v_0 \right) s_y /f  \\
		\end{matrix} \right]. 
	\end{aligned}
\end{aligned}
\ee

Consider the 2D $xoy$-plane rotation with $3$ LEDs in Figure~\ref{2D plane rotation}(a), where the camera resolution $=2560\times1536$, $z_{i}^w-z_{c}^w=2.265$\ m, $f=2.4$\ mm, and $(u_0,\ v_0)=(1305,\ 774)$ after the rotation angle calibration referring to the experimental configuration in~\cite{8644462} and~\cite{8519633}. Table~\ref{tab.Table1} shows the measurements of $(u_{i}^{p}, v_{i}^{p})$ under $8$ values of $\theta_z$ at point $(150, 50)$~\cite{8644462,8519633}. The singular values and NPEM are shown in Figure~\ref{2D plane rotation}(b) according to Equation~(\ref{equ.Det10}). Such experimental results verify Theorem~$1$ that the NPEM is not sensitive to $\theta_z$ if the transmitter plane and receiver plane are parallel.
\begin{figure}[htbp]
	\centering
	\subfigure[Experimental configuration of the rotation in 2D xoy-plane.]{
		\includegraphics[width=3cm, height=4.2cm]{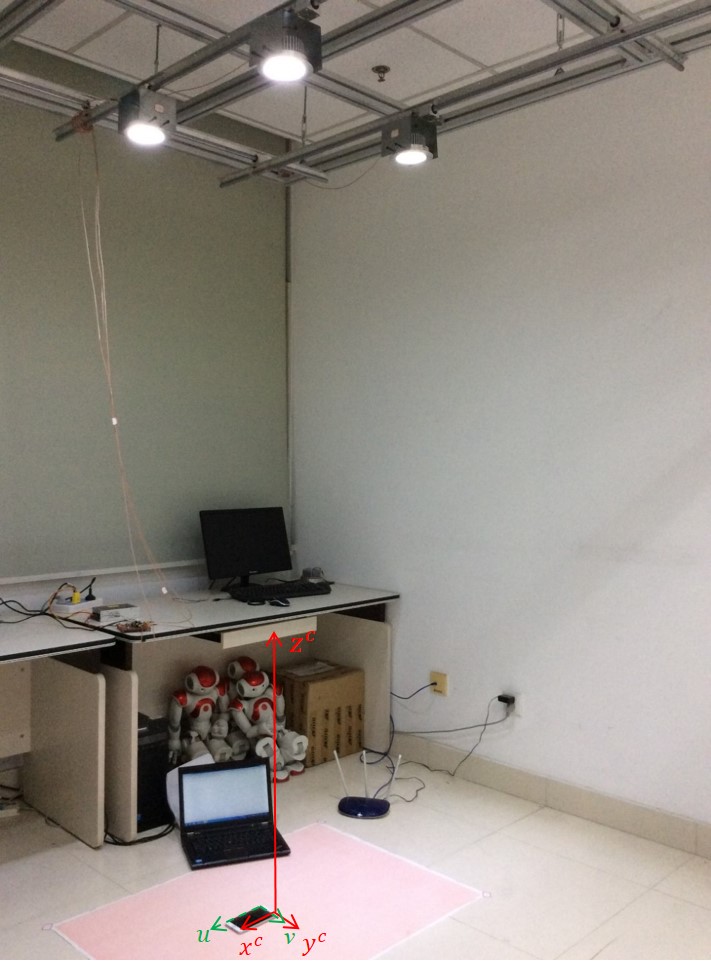}
	}
	\subfigure[The singular values and NPEM with different $\theta_z$.]{
		\includegraphics[width=5.2cm, height=4.2cm]{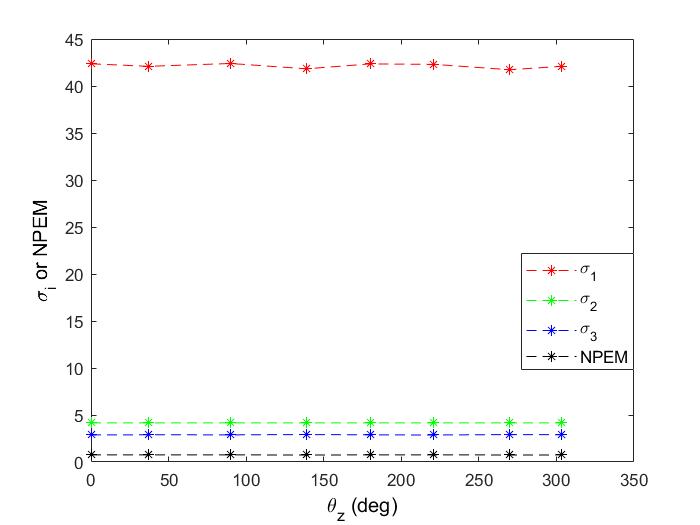}
	}
	\caption{Experimental configuration and NPEM in 2D xoy-plane rotation.}\label{2D plane rotation}
\end{figure}
\renewcommand\arraystretch{1.1}
\begin{table}[htbp]
	\centering
	\caption{The measurements of $(u_{i}^{p}, v_{i}^{p})$ under different rotation angles $\theta_z$}\label{tab.Table1}
	\begin{tabular}{|cc|cc|}
		\hline
		$\theta_z$ (deg) & $(u_{i}^{p}, v_{i}^{p})$ & $\theta_z$ (deg) & $(u_{i}^{p}, v_{i}^{p})$\\
		\hline
		\multirow{3}{*}{0} & (1363.5, 810.5)  & \multirow{3}{*}{180} & (1248.5, 738.5) \\
		& (1373.5, 170.5)  &  &(1237.5, 1377.5) \\
		& (1905.5, 472.5)  &  &(700.5, 1067.5) \\
		\hline
		\multirow{3}{*}{37} & (1328.5, 836.5)  & \multirow{3}{*}{221} & (1283.5, 711.5) \\
		& (1715.5, 325.5)  &  &(860.5, 1188.5) \\
		& (1957.5, 895.5)  &  &(654.5, 613.5) \\
		\hline
		\multirow{3}{*}{90} & (1261.5, 837.5)  & \multirow{3}{*}{270} & (1338.5, 732.5) \\
		& (1907.5, 842.5)  &  &(699.5, 727.5) \\
		& (1604.5, 1376.5)  &  &(1003.5, 192.5) \\
		\hline
		\multirow{3}{*}{139} & (1235.5, 778.5)  & \multirow{3}{*}{303} & (1369.5, 749.5) \\
		& (1655.5, 1258.5)  &  &(842.5, 393.5) \\
		& (1046.5, 1386.5)  &  &(1388.5, 107.5) \\
		\hline	
	\end{tabular}
\end{table}

\section{The relationship between NPEM and LED cell layout}\label{sec.4}
We consider parallel transmitter plane and receiver plane, i.e., $\{\theta_x=\theta_y =0, \theta_z \in [0,2\pi) \}$, which is commonly adopted for real terminal positioning. We further consider the NPEM under square, hexagonal, triangular cell layouts, as shown in Figure~\ref{3 LED layouts}, where each LED is located at the centroid of one cell. The LED density of the three cell layouts is one LED per square meter, such that the LED spacing distances of square, hexagonal and triangular cells are $1$\ m, $\sqrt{2/\sqrt{3}}$\ m and $2\sqrt{1/(3\sqrt{3})}$\ m, respectively. 
\begin{figure}[htbp]
	\centering
	\includegraphics[width=8.5cm, height=3cm]{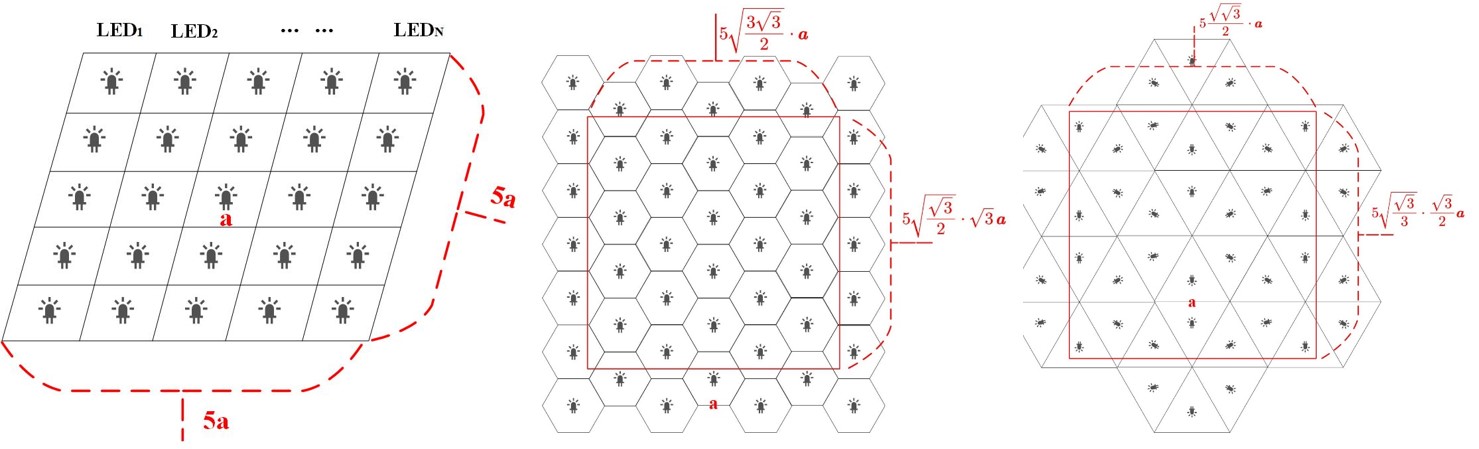}
	\caption{Square, hexagonal and triangular transmitter cell layouts.}\label{3 LED layouts}
\end{figure}
\subsection{The Relationship between NPEM and the Number of Captured LEDs}\label{22}
We explore the relationship between NPEM and the number of captured LEDs under three cell layouts for dense receiver points at different receiver heights.
 
For parallel transmitter plane and receiver plane (i.e., $\theta_x,\theta_y=0$),
\begin{figure}[htbp]
	\centering
	\includegraphics[width=8.8cm, height=3.7cm]{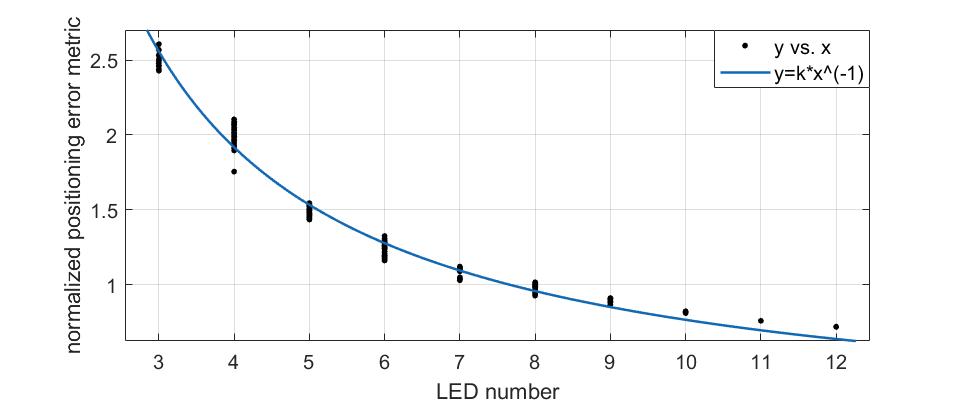}
	\caption{The fitting result of NPEM versus the number of captured LEDs at receiver height $h=1$\ m for square cell layout.}\label{para 3 layouts thetaz=0}
\end{figure}
we obtain the NPEM along with the number of captured LEDs at receiver points $x_{c}^w\in(0:0.1:2.5),\ y_{c}^w\in(0:0.1:5)$ and $z_{c}^w=h$. The fitting result of the NPEM with respect to the number of captured LEDs at $h=1$\ m for square cell layout is shown in Figure~\ref{para 3 layouts thetaz=0}. It can be seen that the NPEM and the number of captured LEDs under square cell layout can be well approximated by relationship $y=kx^{-1}$. Similar results can be obtained for hexagonal and triangular cell layouts. The relationship between NPEM and the number of captured LEDs for all the three cell layouts at different receiver heights is shown in Figure~\ref {para 2layouts thetaz=0}.
\begin{figure}[htbp]
	\centering
	\subfigure[Square cell layout.]{
		\includegraphics[width=7cm, height=5cm]{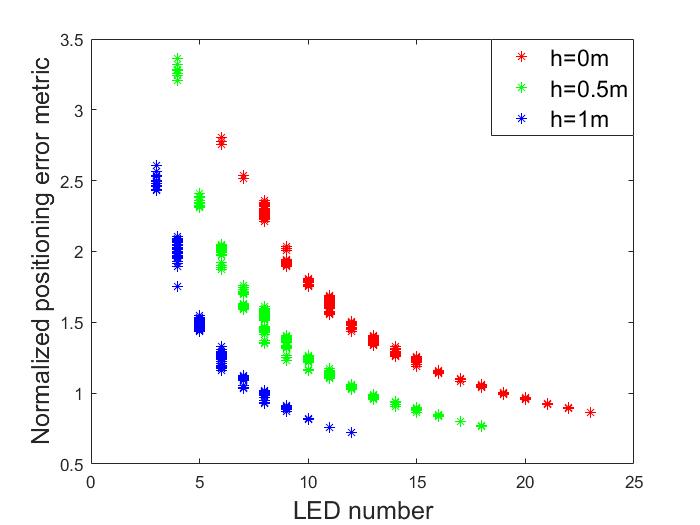}
	}
	\subfigure[Hexagonal cell layout.]{
		\includegraphics[width=7cm, height=5cm]{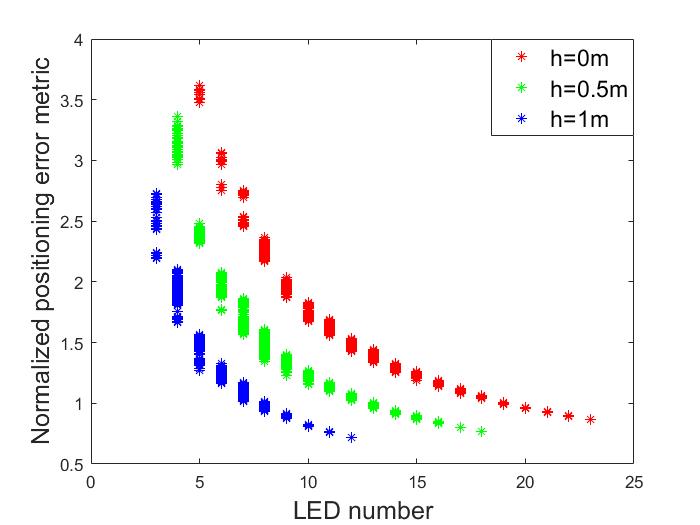}
	}
	\subfigure[Triangular cell layout.]{
		\includegraphics[width=7cm, height=5cm]{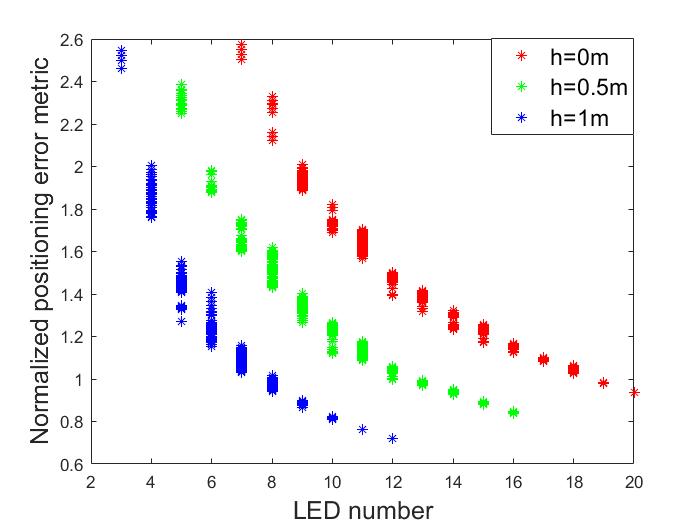}
	}
	\caption{The NPEM versus the number of captured LEDs at different receiver heights for the three cell layouts.}\label{para 2layouts thetaz=0}
\end{figure}

Table~\ref{tab.Table6} shows the fitting parameters and the approximation accuracy at receiver heights of $0$\ m, $0.5$\ m and $1$\ m for the three cell layouts. The determination coefficient R-square represents the quality of data fitting, with the range of $[0, 1]$, the closer it is to $1$, the better the fitting quality is. It can be seen that the NPEM is approximately in inverse proportion to the number of captured LEDs, with close fitting coefficients for all the three cell layouts given receiver height $h$. For a certain layout, take the square cell layout as example, it can also be seen that coefficient $k$ decreases as receiver height $h$ increases, due to shorter transceiver distance that leads to lower NPEM under the same number of captured LEDs.
\begin{table*}[htbp]
*	\centering
	\caption{The fitting results of NPEM with respect to the number of captured LEDs at different receiver heights $h$ for the three cell layouts}\label{tab.Table6}
	\begin{tabular}{|c|c|c|c|c|c|}
		\hline
		Receiver Height (m) & Fitting Curve & Layouts & Coefficient & R-square & RMSE \\
		\hline
		\multirow{3}{*}{$h$=0} & \multirow{9}{*}{$y=kx^{-1}$} & Square & $k$=18.07 & 0.9861 & 0.04647 \\
		\cline{3-6}
		&  & Hexagon & $k$=17.95 & 0.9907 & 0.04803 \\
		\cline{3-6}
		&  & Triangle &  $k$=17.91 & 0.9847 & 0.03898	 \\
		\cline{1-1}\cline{3-6}
		\multirow{3}{*}{$h$=0.5} &  & Square & $k$=12.19 & 0.9763 & 0.07045 \\
		\cline{3-6}
		&  & Hexagon & $k$=12.13 & 0.9864 & 0.06301 \\
		\cline{3-6}
		&  & Triangle & $k$=12.14	& 0.9693 & 0.05508 \\
		\cline{1-1}\cline{3-6}
		\multirow{3}{*}{$h$=1} & & Square & $k$=7.671 & 0.9799 & 0.05631 \\
		\cline{3-6}
		&  & Hexagon &  $k$=7.562 & 0.9712 & 0.07343   \\
		\cline{3-6}
		&  & Triangle & $k$=7.554	 & 0.9613 & 0.05956	 \\
		\cline{1-1}\cline{3-6}
		\hline	
	\end{tabular}
\end{table*}

We can obtain the mean NPEM at different receiver heights for the three cell layouts from Figure~\ref{para 2layouts thetaz=0}, which is shown in Table~\ref{tab.Table10}. It can also be seen that under the density of one LED per square meter, the triangular cell layout with smaller LED spacing leads to lower NPEM, due to more possibility of capturing more LEDs, while hexagonal cell layout shows higher NPEM due to larger LED spacing. However, the hexagonal cell layout is more preferred for communication due to weaker inter-cell interference. Thus, the optimal layout for communication and positioning may not be perfectly aligned, which raises another problem on the joint design of transmitter layout for communication and positioning.
\begin{table}[htbp]
	\centering
	\caption{The mean NPEM at different receiver heights $h$ for the three cell layouts}\label{tab.Table10}
	\begin{tabular}{|c|c|c|}
		\hline
		Layouts  & Receiver Height $h$ (m) & Mean NPEM \\
		\hline
		\multirow{3}{*}{Square}  & 0 & 1.4193  \\
		\cline{2-3}
		& 0.5 & 1.3260  \\
		\cline{2-3}
		& 1 &  1.2089 	 \\
		\cline{1-1}\cline{2-3}
		\multirow{3}{*}{Hexagon} &  0 & 1.5625  \\
		\cline{2-3}
		& 0.5 & 1.4533  \\
		\cline{2-3}
		& 1 & 1.2986	\\
		\cline{1-1}\cline{2-3}
		\multirow{3}{*}{Triangle} & 0 & 1.4091 \\
		\cline{2-3}
		& 0.5 &  1.2514   \\
		\cline{2-3}
		& 1 & 1.1315	 \\
		\cline{1-1}\cline{2-3}
		\hline	
	\end{tabular}
\end{table}

\subsection{The Relationship between Simulated Positioning Error and the Number of Captured LEDs for Square Cell Layout}\label{222}
To show that the inverse proportion approximation relationship $y=kx^{-1}$ is valid for camera-based positioning, we carry out simulations and figure out the relationship between the simulated positioning error and the number of captured LEDs, under square cell layout for parallel transmitter plane and receiver plane. Considering that $\theta_z$ does not affect the NPEM, we set $\theta_z=0$ such that rotation matrix ${\bf R}={\bf I}_3$. We denote the coordinate estimate of user position as $(\hat x^w_c, \hat y^w_c, \hat z^w_c)$, and define the positioning error as $\sqrt{(x_c^w-\hat x^w_c)^2+(y_c^w-\hat y^w_c)^2+(z_c^w-\hat z^w_c)^2}$. 

Based on the transformations among WCS, CCS, ICS and PCS, we have
\be\label{equ.Det906}
\begin{aligned}
 (u_{i}^{p}-u_0 )s_{x}=-f\cdot \frac{(x_{i}^{w}-x_{c}^{w})}{(z_{i}^{w}-z_{c}^{w})},\\
 (v_{i}^{p}-v_0 )s_{y}=-f\cdot \frac{(y_{i}^{w}-y_{c}^{w})}{(z_{i}^{w}-z_{c}^{w})},
\end{aligned}
\ee
for ${\bf R}={\bf I}_3$ without pixel quantization error. Then, with pixel-domain quantization, we adopt the following for terminal positioning,
\be\label{equ.Det20}
\begin{aligned}
	\left( \textbf{Q}[u^p_i]-u_0 \right)s_{x}+f\cdot \frac{\left( x_{i}^{w}-\hat x^w_c \right)}{\left( z_{i}^{w}-\hat z^w_c \right)} = 0, \\
	\left(\textbf{Q}[v^p_i]-v_0 \right)s_{y}+f\cdot \frac{\left( y_{i}^{w}-\hat y^w_c \right)}{\left( z_{i}^{w}-\hat z^w_c \right)} = 0.
\end{aligned}
\ee

Given $z^w_c$, we can obtain the estimates of $\hat x^w_c$ and $\hat y^w_c$ by averaging the solutions of Equation~(\ref{equ.Det20}) for the coordinates $\{(x_i^w, y_i^w, z_i^w)\}$ of the captured LEDs, i.e.,
\be\label{equ.Det22}
\begin{aligned}
\hat x^w_c=\frac{\underset{i}{\boldsymbol{\Sigma}}\left[ x_{i}^{w}-(\textbf{Q}[u^p_i]-u_0) \frac{s_x\left( z_{i}^{w}- z^w_c \right)}{f} \right]}{n},\\ 
\hat y^w_c=\frac{\underset{i}{\boldsymbol{\Sigma}}\left[ y_{i}^{w}-(\textbf{Q}[v^p_i]-v_0) \frac{s_y\left(z_{i}^{w}-z^w_c\right)}{f} \right]}{n}.
\end{aligned}
\ee

We define the following square-based distortion
$$D(\hat z^w_c)= \frac{1}{n}\underset{i}{\boldsymbol{\Sigma}}\left( \left[x_{i}^{w}-(\textbf{Q}[u^p_i]-u_0) \frac{s_x\left( z_{i}^{w}-\hat z^w_c \right)}{f}-\hat x^w_c\right]^2 \right.$$ $$\left.+ \left[ y_{i}^{w}-(\textbf{Q}[v^p_i]-v_0) \frac{s_y\left( z_{i}^{w}-\hat z^w_c\right)}{f}-\hat y^w_c\right]^2 \right).$$
Height $\hat z^w_c$ can be estimated as
\be\label{equ.Det907}
\hat z^w_c=\ \arg \underset{z^w_c}\min \ D(z^w_c). 
\ee 
The above optimization problem can be solved via exhaustive search. Based on the estimated $\hat z^w_c$, we can obtain the estimates of $\hat x^w_c$ and $\hat y^w_c$ based on Equation~(\ref{equ.Det22}). 

We investigate the simulated positioning error under different pixel sizes. Specifically, we let $s_x = s_y=1.675\times 10^{-3}$ under square cell layout constraint~\cite{8644462,8519633}, take the values of pixel size from $0.5s_x$ to $5s_x$ with $50$ values in equal steps. The simulated positioning error versus the number of captured LEDs at different receiver heights is shown in Figure~\ref {5hmeancf}(a). Figure~\ref {5hmeancf}(b) shows the mean simulated positioning error versus the number of captured LEDs.
\begin{figure}[htbp]
	\centering
	\subfigure[Simulated positioning error versus the number of captured LEDs.]{
		\includegraphics[width=7cm, height=5cm]{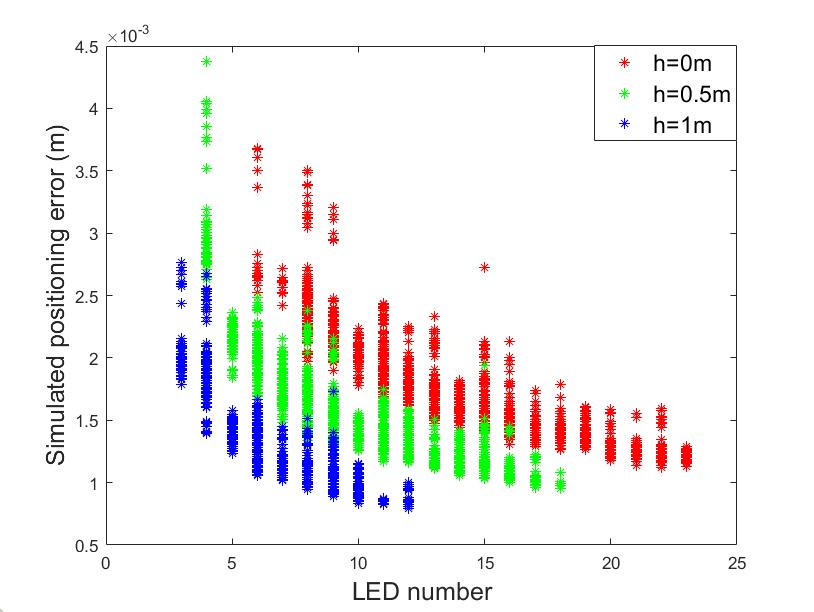}
	}
	\subfigure[Mean simulated positioning error versus the number of captured LEDs.]{
		\includegraphics[width=7cm, height=5cm]{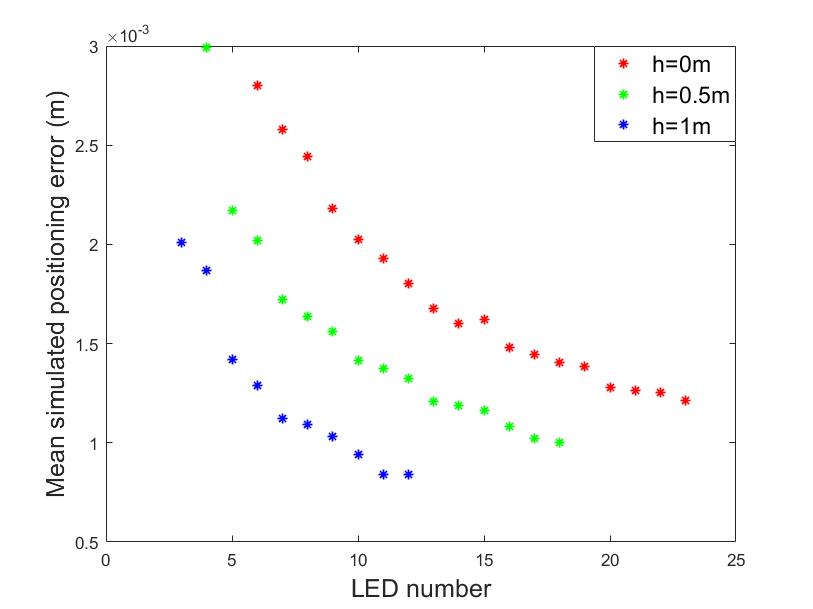}
	}
	\caption{The relationship between the simulated positioning error and the number of captured LEDs at different receiver heights for square cell layout.}\label{5hmeancf}
\end{figure}

Table~\ref{tab.Table7} shows the fitting parameters and the approximation accuracy between the mean simulated positioning error and the number of captured LEDs at receiver heights of $0$\ m, $0.5$\ m and $1$\ m, respectively. 
\begin{table*}[htbp]
	\centering
	\caption{The fitting results of mean simulated positioning error with respect to the number of captured LEDs at different heights $h$}\label{tab.Table7}
	\begin{tabular}{|c|c|c|c|c|}
		\hline
		Receiver Height (m) & Fitting Curve & Coefficient & R-square & RMSE\\
		\hline
		$h$=0 & \multirow{3}{*}{$y=kx^{-1}+c$} & $k$=0.01353, $c$=0.0006536 & 0.9914 & 4.65e-05 \\
		\cline{1-1} \cline{3-5}
		$h$=0.5 & & $k$=0.009273, $c$=0.000498 & 0.9792 & 8.016e-05  \\
		\cline{1-1} \cline{3-5}
		$h$=1  &  & $k$=0.005026, $c$=0.00044 & 0.9688 & 7.677e-05 \\ 
		\hline	
	\end{tabular}
\end{table*}

It can be seen that the mean simulated positioning error and the number of captured LEDs can be approximated by relationship $y=kx^{-1}+c$, which verifies the fitting results of NPEM in Sec.~\ref{222} except a constant term. Such constant can be attributed to the approach in Equations~(\ref{equ.Det22}) and (\ref{equ.Det907}) and numerical approach in the positioning simulation, but can validate the objective of minimizing the NPEM. Coefficient $c$ can be justified by the non-perfect solution of the proposed positioning apporach, but the inverse relationship between the NPEM and the number of captured LEDs can still be verified. Moreover, fitting coefficient $k$ decreases as receiver height $h$ increases, i.e., shorter transceiver distance leads to lower simulated positioning error under the same number of captured LEDs. 

\subsection{The NPEM in Infinite Space for Three Cell Layouts}
We further explore the NPEM in infinite space under both cases where the transmitter plane and receiver plane are parallel and non-parallel, i.e., $\{\theta_x, \theta_y =0, \theta_z \in [0,2\pi] \}$ and $\{ \theta_x,\theta_y\ \in (-\frac{\pi}{2},\ \frac{\pi}{2}) ,\ \theta_z \in [0,2\pi)\}$, respectively. 

Figure~\ref{infinite 3 layouts} shows square, hexagonal and triangular cell layouts with density one LED per square meter in infinite space. According to the illustration of captured LEDs in Figure~\ref{lednumbercaptured}, the LEDs enclosed by red circle and green ellipse include the captured LED under parallel and non-parallel transmitter plane and receiver plane, respectively. 
\begin{figure*}[htbp]
	\centering
	\includegraphics[width=13cm, height=3.6cm]{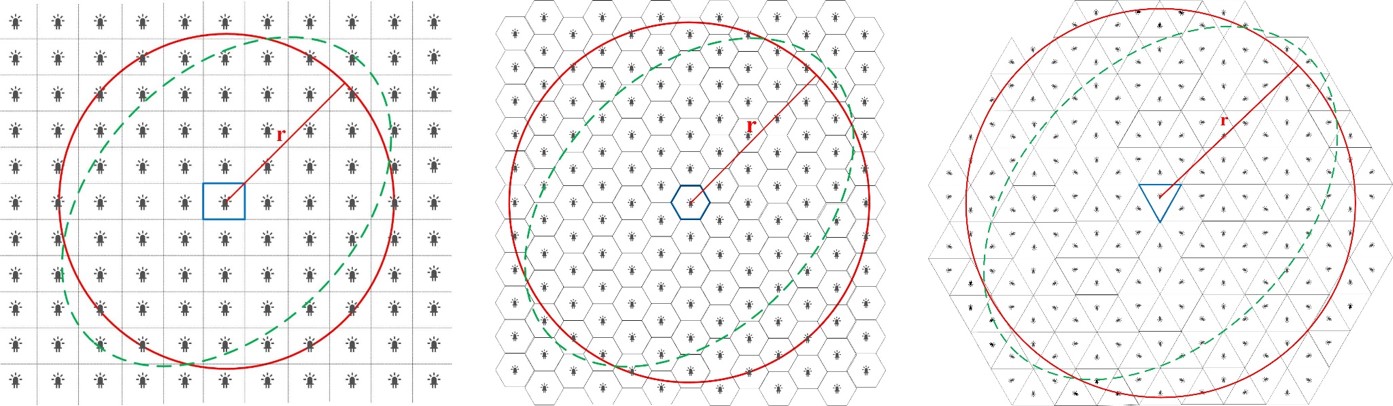}
	\caption{Infinite space of square, hexagonal and triangular cell layouts under parallel and non-parallel transmitter plane and receiver plane (Red circle for parallel case and green ellipse for non-parallel case).}\label{infinite 3 layouts}
\end{figure*}

\subsubsection{Parallel transmitter plane and receiver plane}
For parallel transmitter plane and receiver plane, we explore the relationship between NPEM and the number of captured LEDs at receiver locations $x_{c}^w\in(0:0.01:0.5),\ y_{c}^w\in(0:0.01:0.5)$ and $z_{c}^w=h$ for different receiver heights. Figure~\ref{infinite, para 3layouts} shows the NPEM versus the number of captured LEDs at different receiver heights under the three cell layouts. The fitting results for the three cell layouts are shown in Table~\ref{tab.Table8}.  
\begin{figure*}[htbp]
	\centering
	\subfigure[Square cell layout.]{
		\includegraphics[width=5.5cm, height=4cm]{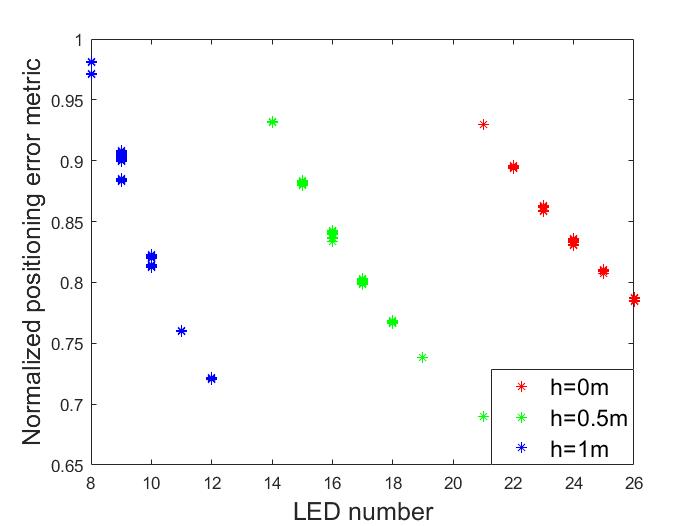}
	}
    \hspace{-2ex}
	\subfigure[Hexagonal cell layout.]{
		\includegraphics[width=5.5cm, height=4cm]{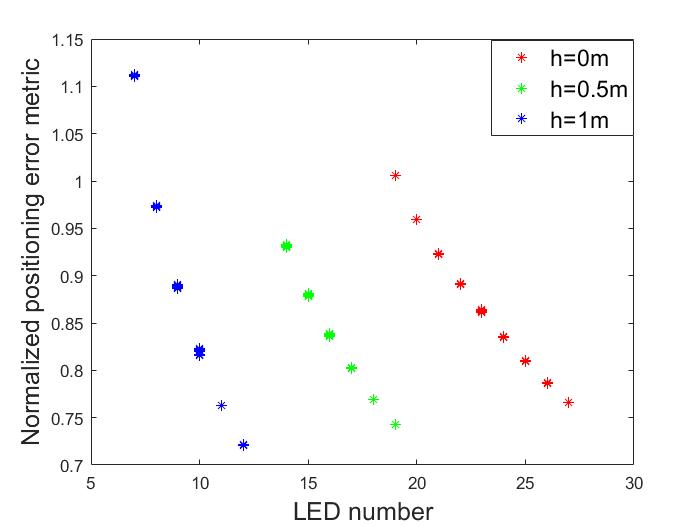}
	}
	\hspace{-2ex}
	\subfigure[Triangular cell layout.]{
		\includegraphics[width=5.5cm, height=4cm]{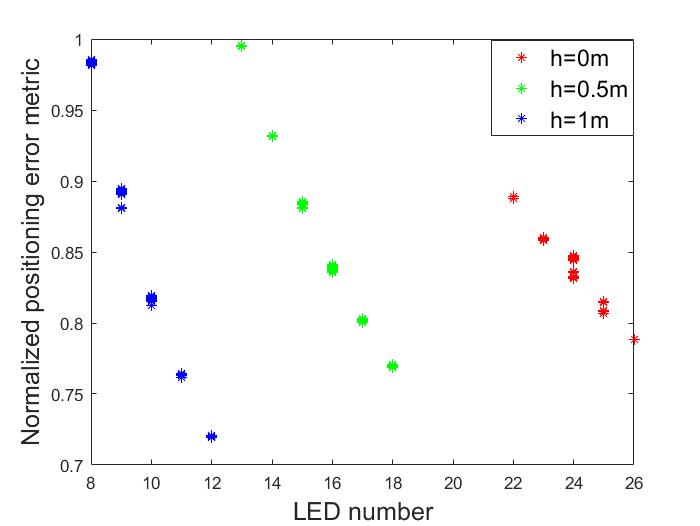}
	}
	\caption{The NPEM versus the number of captured LEDs for infinite space at different receiver heights.}\label{infinite, para 3layouts}
\end{figure*}
\begin{table*}[htbp]
	\centering
	\caption{The fitting results of NPEM with respect to the number of captured LEDs in infinite space at different receiver heights $h$ under parallel transmitter plane and receiver plane}\label{tab.Table8}
	\begin{tabular}{|c|c|c|c|c|c|}
		\hline
		Receiver Height (m) & Fitting Curve & Layouts & Coefficient & R-square & RMSE \\
		\hline
		\multirow{3}{*}{$h$=0} & \multirow{9}{*}{$y=kx^{-1}+c$} & Square & $k$=15.63, $c$=0.1834 & 0.9986 & 0.001251 \\
		\cline{3-6}
		&  & Hexagon & $k$=15.36, $c$=0.1947 & 0.9994 & 0.000931 \\
		\cline{3-6}
		&  & Triangle & $k$=14.35, $c$=0.2366 & 0.9991 & 0.001397	 \\
		\cline{1-1}\cline{3-6}
		\multirow{3}{*}{$h$=0.5} &  & Square & $k$=10.35, $c$=0.1923 & 0.9983 & 0.001598 \\
		\cline{3-6}
		&  & Hexagon & $k$=10.02, $c$=0.2118 & 0.999 & 0.001531 \\
		\cline{3-6}
		&  & Triangle & $k$=10.47, $c$=0.1849 & 0.9988 & 0.001507 \\
		\cline{1-1}\cline{3-6}
		\multirow{3}{*}{$h$=1} & & Square & $k$=6.295, $c$=0.1919 & 0.983 & 0.008044 \\
		\cline{3-6}
		&  & Hexagon &  $k$=6.441, $c$=0.1751 & 0.9957 & 0.005172   \\
		\cline{3-6}
		&  & Triangle & $k$=6.414, $c$=0.1789 & 0.9986 & 0.002877	 \\
		\cline{1-1}\cline{3-6}
		\hline	
	\end{tabular}
\end{table*}

Similar to the results in $5$m$\times5$m$\times3$m finite space, the NPEM is approximately in inverse proportion to the number of captured LEDs at different receiver heights. Coefficient $k$ at different $h$ decreases as the distance between the transmitter plane and the receiver plane decreases. 
\subsubsection{Non-parallel transmitter plane and receiver plane}
We investigate the NPEM with different rotation angles under non-parallel transmitter plane and receiver plane with $\{ \theta_x,\theta_y\ \in (-\frac{\pi}{2},\ \frac{\pi}{2}) ,\ \theta_z \in [0,2\pi)\}$, assuming that the receiver point locates at $x_{c}^w\in(0:0.2:0.5),\ y_{c}^w\in(0:0.2:0.5)$ and $z_{c}^w=h$ under different receiver heights.

We can obtain the minimum NPEM at each receiver point through traversing $\theta_x,\theta_y\ \in (-\frac{\pi}{2},\ \frac{\pi}{2})$ with step size of $\pi/9$ and $\theta_z \in [0,2\pi)$ with step size of $\pi/2$. There is no explicit quantitative fitting relationship between the NPEM and the number of captured LEDs, since the optimized rotation angle critically depends on the receiver position. 

\section{Indoor dense LED layout optimization}\label{sec.5}
\subsection{Optimization Problem Formulation}
Based on the fitting relationship between NPEM and the number of captured LEDs under different cell layouts, we investigate the LED cell layout optimization under parallel transmitter plane and receiver plane.
\begin{figure}[H]
	\centering
	\includegraphics[width=4cm, height=4.5cm]{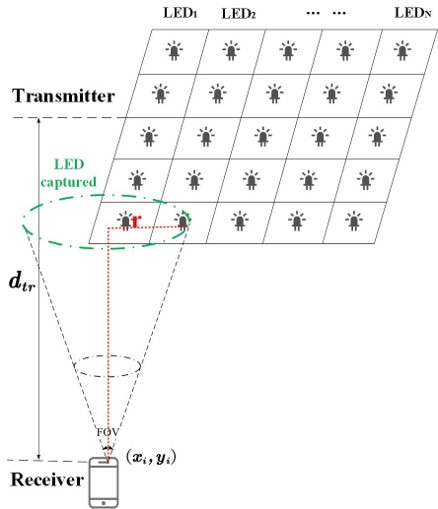}
	\caption{The LEDs captured by camera with parallel transmitter plane and receiver plane.}\label{LED layout optimal}
\end{figure}

We optimize the LED cell layout based on the NPEM of $I$ points. Let $n$ be the number of transmitter LEDs, and $n_{c}^{i}$ be the number of captured LEDs at receiver $(x_i,y_i)$, for $i=1,2,...,I$, which can be obtained according to the illustration of captured LEDs in Sec.~\ref{3D error perf}. Figure~\ref{LED layout optimal} shows the coverage radius $r$, given by $r=d_{tr}\cdot \tan (\frac{FOV}{2})$, where $d_{tr}$ is the distance between the transmitter and the receiver, and $n_{c}^{i}$ is the number of LEDs inside the circle with center $(x_i,y_i)$. We optimize the LED cell layout to minimize the mean NPEM over all the $I$ points, via optimizing the transmitter LED coordinates $(x_j^w,y_j^w,z_j^w)$, for example with constant height $z_j^w=2.75$\ m for $j=1,2,...,n$. Moreover, since from both fitting and simulated position, the positioning error is approximately in inverse proportion to the number of captured LEDs, we convert minimizing NPEM into minimizing the mean of $1/n_{c}^{i}$, i.e.,
\be\label{equ.Det19}
\underset{\left\{ \left( x_{j}^{w},\ y_{j}^{w} \right) \right\}_{j=1}^{n}}\min \  \frac{1}{I}\sum_{i=1}^I{\frac{1}{n_{c}^{i}}}.
\ee 

\subsection{Rectangular Cell Layout Optimization}
We investigate the NPEM under non-uniform $M\times N$ rectangular cell layout considering common indoor room, as shown in Figure~\ref{non-uniform M*N layout}. Such layout yields rectangular pattern but the spacings between adjacent rows and columns can be designed. The row and column spacings, denoted as $d_{r,i}\ ( i=1,2,3,...,M+1)$ and $d_{c,j}\ (j=1,2,3,...,N+1)$, respectively, are symmetric, i.e., row spacing $d_{r,i}=d_{r,M+2-i}$ and column spacing $d_{c,j} = d_{c,N+2-j}$.
\begin{figure}[htbp]
	\centering
	\includegraphics[width=5cm, height=4.3cm]{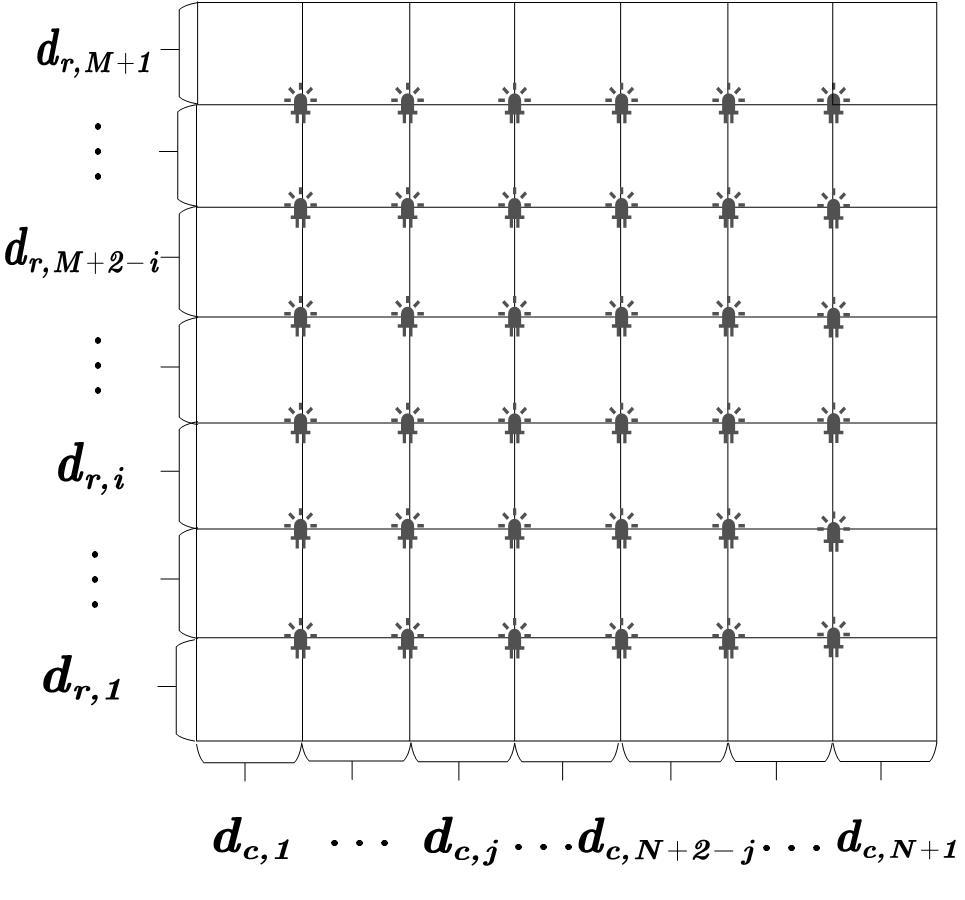}
	\caption{Non-uniform $M\times N$ rectangular cell layout.}\label{non-uniform M*N layout}
\end{figure}

Firstly, we consider the NPEM of uniform $5\times 5$ square cell layout as an example, i.e., $M=N=5, d_{r,i}=d_{c,j}$, spacing $d_{r,i}=d_{c,j}=0.5$\ m, $i,j=1$ and $d_{r,i}=d_{c,j}=1$\ m, $i,j=2,3$, as shown in Figure~\ref{uniform square}(a). The NPEM is shown in Figure~\ref{uniform square}(b), where $x$- and $y$- axis represent the receiver points. The mean NPEM is $0.5831$ for all the receiver $( x_i,y_i) \in \{ [ 0:1:5] ,\ [ 0:1:5]\}$. 
\begin{figure}[htbp]
	\centering
	\subfigure[Uniform square cell layout.]	{
		\includegraphics[width=6cm, height=5cm]{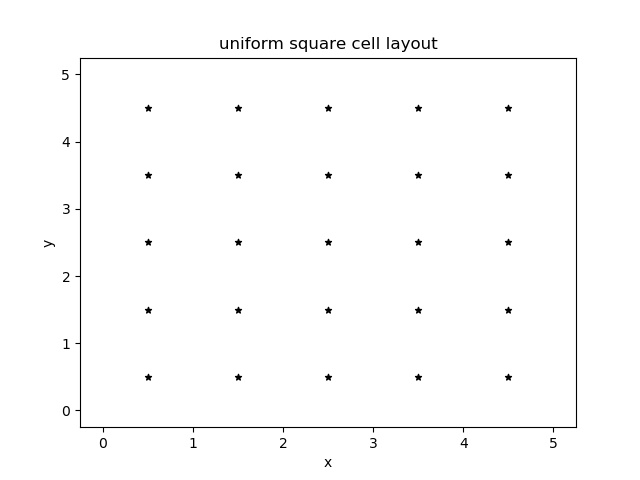}
	}
	\subfigure[NPEM of uniform square cell layout.]{
		\includegraphics[width=6cm, height=5cm]{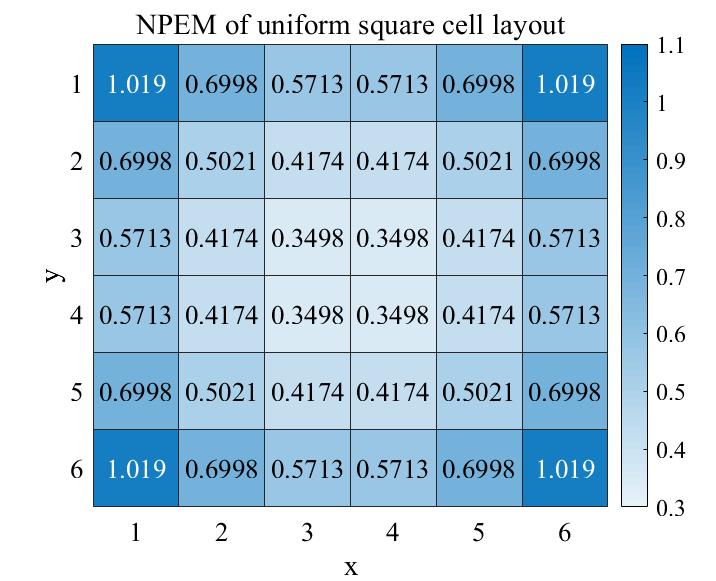}
	}
	\caption{NPEM under uniform square cell layout.}\label{uniform square}
\end{figure}

Then, we optimize the NPEM under non-uniform $5\times 5$ rectangular cell layout, i.e., $M=N=5$. The row spacing $d_{r,i}$ and column spacing $d_{c,j}$ are symmetric, respectively, as shown in Figure~\ref{non-uniform square layout}(a). Figure~\ref{non-uniform square layout}(b) shows the optimized non-uniform rectangular cell layout with the mean NPEM of $0.5276$, which is lower than that of the uniform square cell layout in Figure~\ref{uniform square}. It is seen that the optimized non-uniform rectangular cell layout shows lower mean NPEM in VLP. While uniform square cell layout is generally preferred in VLC and commonly adopted in real indoor lighting, such result raises a challenge to find a balance between the communication and positioning for the joint design of VLC and VLP. 
\begin{figure}[htbp]
	\centering
	\subfigure[Non-uniform $5\times 5$ rectangular cell layout.]{
		\includegraphics[width=6cm, height=5cm]{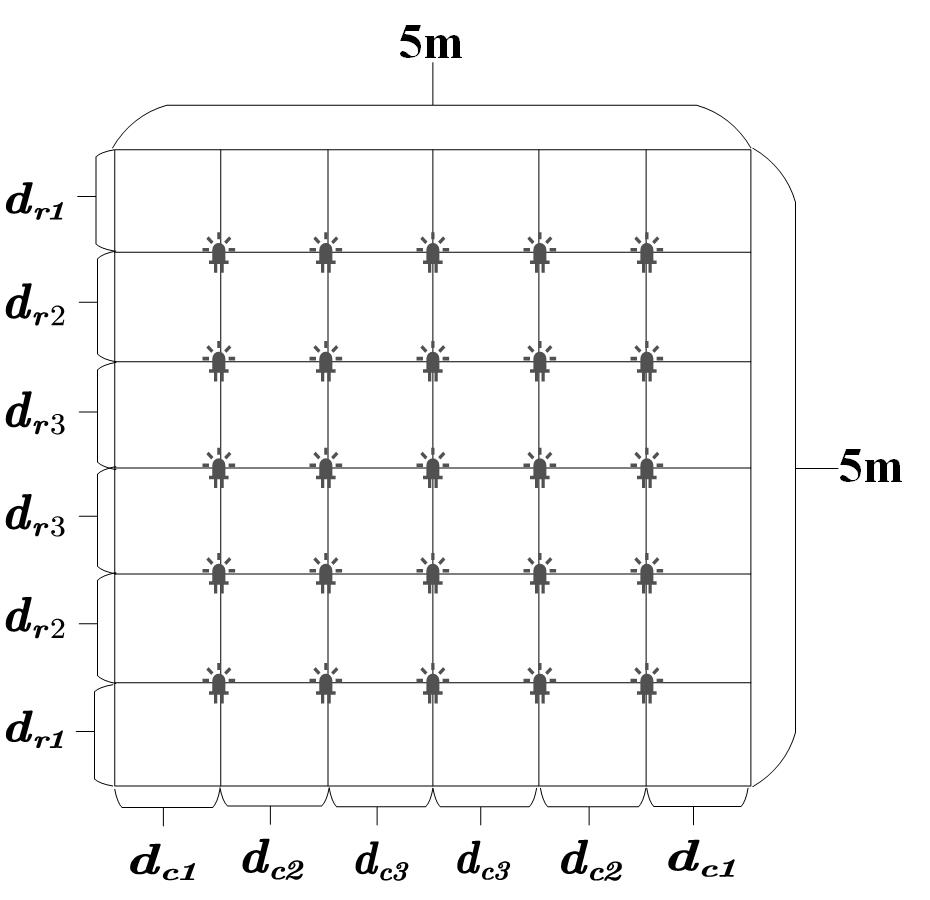}
	}
	\subfigure[Optimized non-uniform $5\times 5$ rectangular cell layout.]{
		\includegraphics[width=5.6cm, height=5cm]{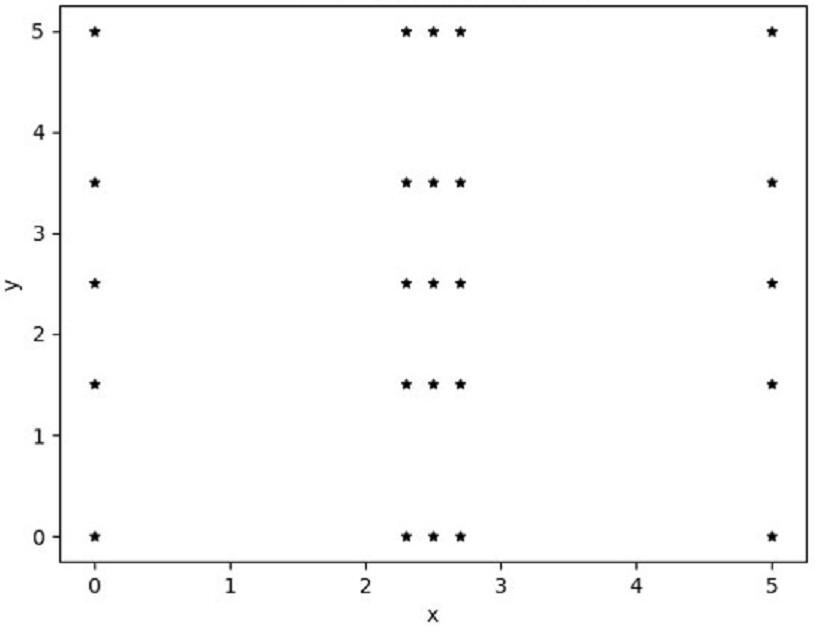}
	}
	\caption{Optimized non-uniform rectangular cell layout.}\label{non-uniform square layout}
\end{figure}

Meanwhile, we solve Problem~(\ref{equ.Det19}) through genetic algorithm (GA), which consists of coding, fitness function, and genetic operators (selection, crossover and variation). We set selection rate to be $0.5$, crossover rate to be $0.7$, variation rate to be $0.001$, while the number of iteration depends on initial population size. In addition, to accelerate the iterative convergence rate and guarantee the convergence to a high fitness value, we consider selecting individuals with high fitness to the following generation. Under selection rate of $0.5$, we select the individuals with the highest half of fitness as the parents in the next generation. Figure~\ref{5060,40} shows the convergence of fitness and corresponding layout for initial population of $5071$ and $10$ iterations.
\begin{figure}[htbp]
	\centering
	\subfigure[The fitness convergence.]{
		\includegraphics[width=7.5cm, height=5cm]{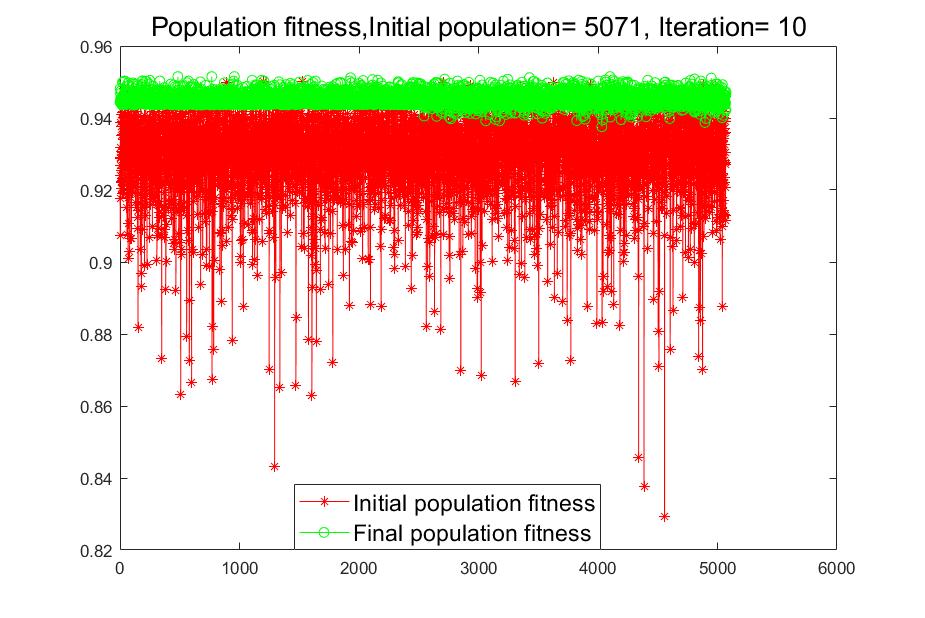}
	}
	\subfigure[Corresponding LED layout.]{
		\includegraphics[width=5.8cm, height=4.6cm]{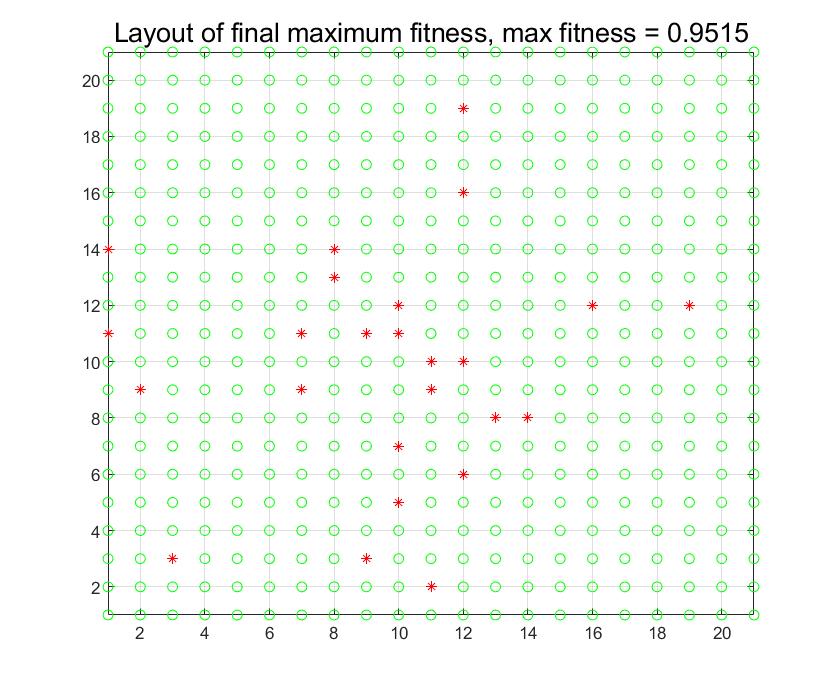}
	}
	\caption{Initial population=$5071$, iterations=$10$.}\label{5060,40}
\end{figure}

Similarly, we can obtain the mean NPEM of $0.5777$ under the optimized layout of GA for the receiver $( x_i,y_i) \in \{ [ 0:1:5] ,\ [ 0:1:5]\}$, which is higher than that obtained under non-uniform rectangular cell layout. It is seen that the GA shows higher mean NPEM compared with the proposed optimized non-uniform rectangular cell layout, due to high dimensions of optimization.

\section{Conclusion} \label{sec.Conclusions}
We consider a 3D VLP based on smartphone camera in the indoor scenario, analyze the 3D positioning NPEM, derive the NPEM expression through the partial derivative of the position relationship, and evaluate the numerical results of the NPEM under horizontal and non-horizontal receiver planes. Moreover, we approximate the relationship between the NPEM and the number of captured LEDs for parallel transmitter plane and receiver plane, and explore the NPEM in infinite LED cell layout. Finally, we optimize the LED transmitter cell layout to minimize the NPEM, and provide parameter optimization under square cell layout. Note that the above NPEM analysis and layout optimization are based on the ideal smartphone camera imaging with negligible distortion, while the positioning with nonlinear distortion need to be further explored in future work.


\bibliographystyle{ieeetr}
\bibliography{./xu_1224}
	
\end{document}